\documentclass[journal]{IEEEtran}

\usepackage{amsmath,amsfonts,amssymb}
\usepackage{graphicx}
\usepackage{bm}
\usepackage{textcomp}
\usepackage{xcolor}
\usepackage{cite}
\newtheorem{lemma}{Lemma}

\newcommand{\paperfigure}[2][]{\IfFileExists{#2}{\includegraphics[#1]{#2}}{\fbox{\parbox[c][1.35in][c]{0.94\linewidth}{\centering Missing figure: \texttt{\detokenize{#2}}}}}}

\newtheorem{proposition}{Proposition}
\newtheorem{remark}{Remark}

\begin{document}
\title{\huge Cold-Atom Rydberg Array with Spatiotemporal Multiplexing for High-Sensitivity Wireless Communications}

\author{Jian Xiao, Tierui Gong, and Chau Yuen,~\IEEEmembership{Fellow,~IEEE}
\thanks{ Jian Xiao, Tierui Gong, and Chau Yuen are with the School of Electrical and Electronics Engineering, Nanyang Technological University, Singapore 639798. (e-mail: jian.xiao@ntu.edu.sg; tierui.gong@ntu.edu.sg; chau.yuen@ntu.edu.sg).}}

\maketitle

\begin{abstract}
An effective Cold-Atom Rydberg Array (CoRA) architecture is investigated, which addresses the fundamental thermal noise limit in classical electromagnetics and intrinsic Doppler broadening as well as photon shot noise in conventional hot-vapor cell-based atomic receivers. However, translating the CoRA-based quantum electrometers into practical wireless communication receivers is bottlenecked by both the macroscopic readout dead-times and nonlinear binomial population counts rather than the additive-Gaussian baseband samples in conventional receivers. To address these critical challenges and facilitate the practicality of the CoRA, in this paper, a subarray scheduling-based spatiotemporal multiplexing (STM) framework is first proposed for gap-free symbol interrogation within an admitted packet. Additionally, a signal detection scheme based on a joint binomial particle smoother (JBPS) is proposed {for} burst-wise state and data inference. Furthermore, the weak-signal quantum-projection-noise (QPN) limit and the state-averaged finite-alphabet rate are derived for CoRA assisted wireless communications. Numerical results demonstrate that: 1) The proposed STM-based CoRA scheme resolves the macroscopic temporal bottleneck for continuous signal sampling; 2) The proposed JBPS approach effectively mitigates the nonlinear noise amplification to achieve a lower bit error rate; and 3) {{Performance analysis reveals a fundamental sensitivity–sustainability tradeoff that increasing the number of atoms effectively suppresses the QPN, but simultaneously reduces the sustainable packet repetition rate.}}

\end{abstract}


\section{Introduction}
\label{sec:introduction}

\IEEEPARstart{H}{igh-sensitivity} reception is a longstanding objective in wireless receivers, as it extends communication range and enables reliable signal recovery under weak incident fields. However, the sensitivity of a conventional antenna-based receiver is fundamentally constrained by thermal noise, i.e., the Johnson--Nyquist noise \cite{Johnson1928Thermal}, generated by its dissipative antenna impedance and receiver electronics. Benefiting from the large electric-dipole moments between neighboring Rydberg states that can produce strong microwave coupling, Rydberg atomic quantum receivers (RAQRs) emerge as a novel receiver technology for detecting the ultra-weak electromagnetic fields that fundamentally replaces classical metallic antennas with quantum sensors composed of Rydberg atom ensembles \cite{10972179, 11124471, anderson2020atomic}. More particularly, most RAQRs demonstrated to date employ room-temperature vapor cells. Despite their experimental accessibility, vapor cell-based RAQRs rely on a thermally moving and ensemble-averaged atomic medium. In this case, the atomic Doppler broadening, collisional decoherence, and spatially nonuniform coupling make their sensitivity, bandwidth, and dynamic range strongly dependent on the excitation geometry and operating point \cite{11692870,glick2025doppler, ref_cui_mimo}.

In parallel, a Cold-Atom Rydberg Array (CoRA) has emerged as a rapidly advancing frontier of quantum science, as the deterministic atom-by-atom assembly of defect-free arrays was realized through single-atom-resolved imaging and optical-tweezer rearrangement \cite{endres2016atom,barredo2016atom}, establishing the individually programmable atomic systems. Currently, the capabilities of cold-atom arrays have enabled major advances in quantum simulation, many-body physics, and logical quantum computing \cite{kaufman2021quantum,bluvstein2024logical,manetsch2025tweezer}. In particular, recent CoRA studies have demonstrated their potential for approaching the standard quantum limit (SQL) \cite{duverger2024metrology,zhang2026microwave, gong2026quantumdomain}. However, existing implementations have primarily focused on estimating the amplitude or phase of controlled, quasi-stationary {microwave (MW)} fields, rather than continuously recovering dynamically modulated, information-bearing waveforms under practical communication constraints. These developments motivate the central question of this work: \emph{\textbf{whether a programmable CoRA can be transformed from a discrete quantum sensor into a reliable wireless receiver when finite-atom statistics, destructively probabilistic readout, and atom-replenishment constraints are jointly taken into account?}}

\subsection{Prior Works}

\begin{figure*}[t]
\centerline{\includegraphics[width=6.8in]{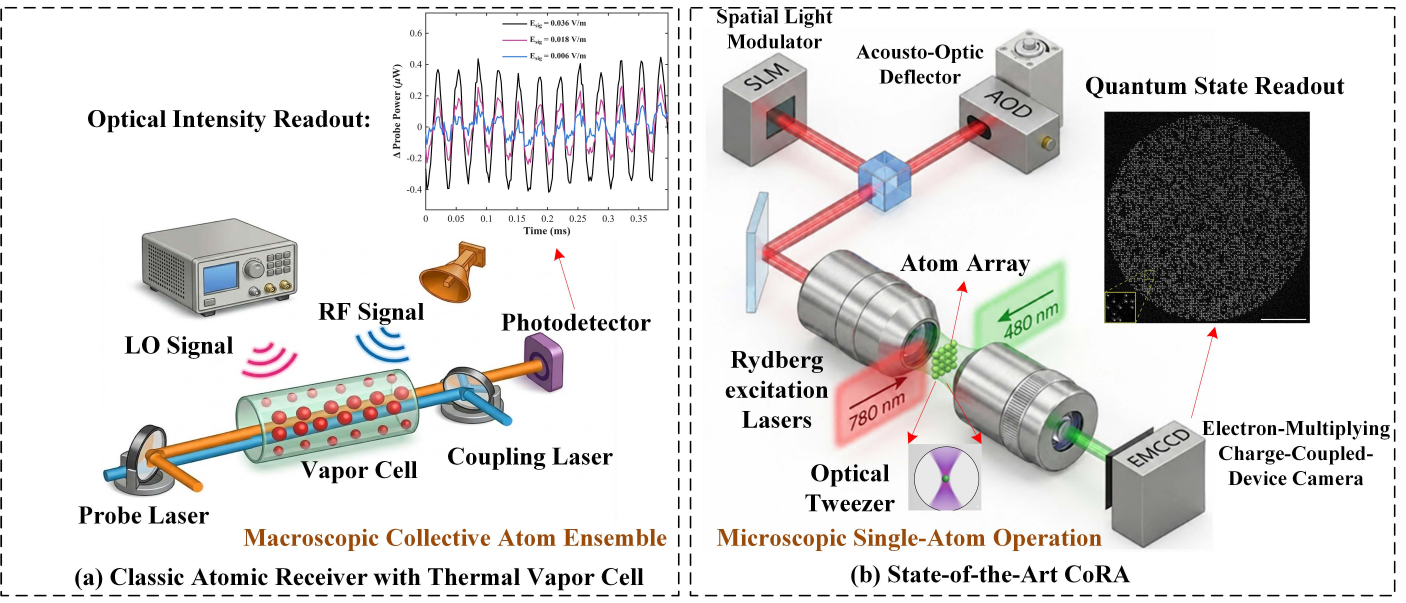}}
\caption{Rydberg atom-based quantum electrometry architectures. (a) The classic atomic receiver utilizing a macroscopic hot vapor cell \cite{meyer2018digital, Simons2019RydbergMixer, holloway2019detecting}, where the microwave field is  measured through the collective optical intensity readout of an atomic ensemble. (b) The state-of-the-art CoRA empowered by optical tweezers \cite{ebadi2021quantum,chiu2025continuous,gyger2024continuous}, which isolates microscopic laser-cooled atoms to enable single-atom quantum state readout via state-selective fluorescence imaging.}
\label{fig:level_schemes}
\end{figure*}

Existing RAQRs can be systematically classified into three dominant frameworks based on their hardware implementation and readout mechanisms. As illustrated in Fig. \ref{fig:level_schemes}(a), the first framework encompasses {macroscopic vapor-cell-based RAQRs}, which leverage the electromagnetically induced transparency (EIT) when interacting with room-temperature atomic ensembles \cite{ ref_gong_overview, 11482843, chen-Rydberg2, wang2026multicarrier}. While capable of high-sensitivity RF sensing, they are fundamentally bottlenecked by severe Doppler broadening \cite{meyer2021optimal, glick2025doppler}, blackbody radiation \cite{cooke1980effects, kaur2025impact}, and a spatial resolution restricted to submillimeter scale as well as limited instantaneous bandwidth \cite{holloway2014sub, hu2023improvement}. The second paradigm utilizes {collective cold-atom ensembles} confined in magneto-optical traps (MOTs) \cite{tu2024approaching, duverger2024metrology}. Although macroscopic cold-atom ensembles significantly improve the absolute sensitivity by suppressing the Doppler broadening and transit-time noise, their temporal resolution remains fundamentally constrained. Since these systems still rely on the continuous establishment of macroscopic EIT dark states, their sensitivity is limited by photon shot noise in optical readout.

To overcome the Doppler and Poissonian limits in conventional vapor-cell-based RAQRs, the state-of-the-art in quantum electrometry has recently shifted to the CoRA architecture \cite{kaufman2021quantum,manetsch2025tweezer,holman2026trapping}. As shown in Fig. \ref{fig:level_schemes}(b), by utilizing tightly focused optical tweezers in an ultra-high vacuum environment, this architecture traps and isolates individual and laser-cooled neutral atoms, which eliminates collision-induced decoherence \cite{browaeys2020many}. Instead of utilizing steady-state collective optical responses, the authors of \cite{zhang2026microwave} employ single-atom homodyne measurements and discrete state-selective fluorescence imaging to approach the SQL for microwave electrometry. This paradigm shift transitions the observation from Poissonian statistics to binomial projection statistics, enabling a field sensitivity bounded strictly by the SQL \cite{adams2019rydberg}. In particular, the single-atom excitation is governed purely by ultra-fast Rabi dynamics rather than slow collective relaxation, inherently facilitating an ultra-fast temporal response. Moreover, unlike classical metallic antennas that suffer from severe mutual impedance coupling when spaced below the half-wavelength limit, laser-cooled neutral atoms are intrinsically immune to macroscopic electromagnetic crosstalk. This lack of mutual coupling permits array lattice spacings at the micrometer scale, directly translating to an extreme subwavelength spatial resolution \cite{zhang2026microwave}.

\subsection{Motivations and Contributions}

Despite these metrological breakthroughs, existing cold-atom arrays are engineered exclusively for discrete-time quantum sensing \cite{zhang2026microwave}, presenting severe physical bottlenecks for continuous-time communications and sensing. Firstly, the extraction of the coherent quantum state relies on destructive fluorescence imaging, causing quantum-state collapses of the CoRA. Although single-atom excitation theoretically exhibits an ultrafast temporal response, translating this advantage into continuous high-speed data streams faces a catastrophic hardware bottleneck in the time domain. The requisite optical pumping, coherent interrogation, imaging, and re-cooling cycle creates a macroscopic dead-time on the order of tens of milliseconds, e.g., limiting the measurement rate to $53$ Hz \cite{zhang2026microwave}. This fundamentally violates the temporal continuity required for megabaud communication symbol streams. Secondly, unlike classical analog receivers, the discrete binary quantum state readout of cold-atom arrays fundamentally degrades continuous-time communication waveforms into highly nonlinear binomial distributions dominated by the severe quantum shot noise. This profound quantization bottleneck invalidates traditional Gaussian-based signal processing, triggering catastrophic noise amplification under limited atomic resources.

Against the above background, we investigate CoRA assisted wireless communication systems and propose a hardware-protocol co-design signal processing scheme. The main contributions are summarized as follows:
\begin{itemize}
\item  We propose a spatiotemporal multiplexing (STM) framework based on the continuous-operation CoRA architecture to overcome the fundamental temporal resolution constraints of discrete-time quantum sensing. By exploiting a steady-state dual-lattice atom supply chain and multiplexing stationary subarrays via optical steering, the STM architecture transforms the CoRA from a duty-cycle-limited sensor into a continuous-time burst-packet receiver tailored for modern wireless networks.

\item We derive an end-to-end probabilistic model that maps transmit power through free-space spreading and coherent atomic interrogation to exact in-phase and quadrature (I/Q) binomial counts, in which a three-dimensional latent state connects the physical phase and Rabi-gain variations to the observation likelihood. Then, we propose a joint binomial particle smoother (JBPS) for burst-wise state and data inference, which exploits all observations in the finite burst before bitwise posterior decisions are formed.

\item We establish a unified performance framework that connects the active atom number to receiver sensitivity, finite-alphabet rate, and long-term packet sustainability. \textbf{\textit{{{In particular, we prove the minimum number of active atoms per subarray required for the CoRA to outperform the specified classical-antenna and hot vapor-cell benchmarks, and reveal a sensitivity–sustainability scaling law governed by the available atom flux of the CoRA.}}}}

\item We demonstrate through numerical results that: 1) The proposed STM overcomes the macroscopic measurement-cycle bottleneck within each admitted burst, thereby enabling gap-free symbol-rate sampling; 2) The proposed JBPS detection avoids the nonlinear noise amplification, thereby achieving superior bit-error-rate performance over the conventional signal detection baselines; and 3) In the quantum projection noise (QPN)-limited and atom-flux-limited regime, the weak-signal SNR linearly scales as the number of atoms per subarray $K_m$, whereas the sustainable packet repetition frequency scales as $K_m^{-1}$, so that their product is independent of $K_m$.
\end{itemize}

\subsection{Organization and Notation}

{In Section II, we introduce the experimentally demonstrated pulsed and continuous-loading Rydberg-atom-array hardware backbones. In Section~\ref{sec:short_burst_receiver}, we propose the STM architectures and the JBPS detection algorithms.} In Section~\ref{sec:performance_analysis_r6}, we derive the information-theoretic performance measures for different antenna architectures. Lastly, we present the numerical evaluation in Section V and conclude the article in Section~\ref{sec:conclusion}.

\emph{Notations:} The superscripts $(\cdot)^T$ and $(\cdot)^H$ stand for the transpose and conjugate transpose operations. $\mathbb{E}[\cdot]$ represents the statistical expectation. $|\cdot|$ and $\|\cdot\|_2$ denote the absolute value of a scalar and the $\ell_2$-norm of a vector, respectively. The operators $\Re\{\cdot\}$ and $\Im\{\cdot\}$ extract the real and imaginary parts of a complex number, respectively, and $\angle(\cdot)$ denotes the phase angle. $|\mathcal{A}|$ represents the cardinality of the set $\mathcal{A}$. The notation $\mathbf{I}_K$ stands for the $K \times K$ identity matrix. $\mathcal{N}(\bm{\mu}, \mathbf{\Sigma})$ denotes a Gaussian distribution with mean vector $\bm{\mu}$ and covariance matrix $\mathbf{\Sigma}$. $\mathbb{I}_{[\cdot]}$ denotes the indicator function, which equals $1$ if the condition holds and $0$ otherwise. $\hbar$ denotes the reduced Planck constant. $\exp(\cdot)$ denotes the exponential function.

\section{Physical Architectures and Implementations of Cold-Atom Arrays}
\label{sec:physical_architecture}

In this section, we first review the experimentally demonstrated pulsed Rydberg-atom-array electrometer and formulate its single-atom measurement model \cite{zhang2026microwave}. We then demonstrate a continuous atom-loading architecture that serves as the hardware backbone of CoRA systems.

\subsection{Pulsed CoRA for Microwave Electrometry}

\begin{figure}[t]
\centering \includegraphics[width=3.0in]{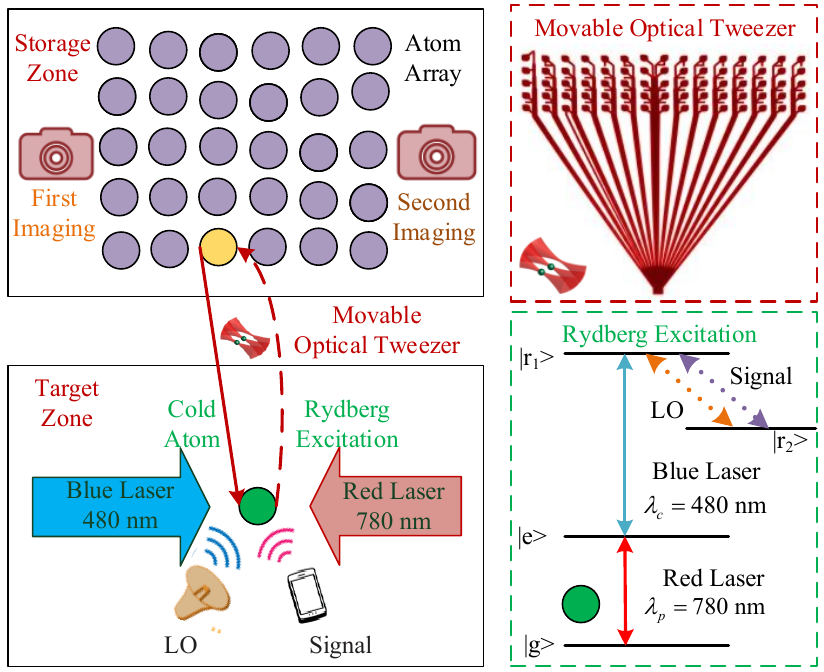}
\caption{{Existing} pulsed Rydberg-array electrometer based on the single movable optical tweezer \cite{zhang2026microwave}.
}
\label{fig:classic_array}
\end{figure}

As shown in Fig.~\ref{fig:classic_array}, the pulsed architecture in \cite{zhang2026microwave} provides site-resolved preparation, control, and readout of individual $^{87}\mathrm{Rb}$ atoms. Compared with a thermal vapor cell, single-atom isolation suppresses collision-induced decoherence, although atomic motion, laser noise, trap inhomogeneity, spontaneous decay, and blackbody-radiation coupling remain relevant.

\subsubsection{Dual-Zone Optical-Tweezer Architecture}

The pulsed platform contains two functional zones.
\begin{itemize}
\item \emph{Reservoir zone}: Laser-cooled atoms are probabilistically loaded into a static two-dimensional tweezer array, and fluorescence imaging identifies the occupied sites.
\item \emph{Target zone}: A movable tweezer transports selected atoms to a spatially separated region for state preparation, Rydberg excitation, MW interrogation, and state-selective de-excitation.
\end{itemize}
In the target zone, an atom is optically pumped to $|g\rangle=|5S_{1/2},F=2,m_F=2\rangle$ and transferred by stimulated Raman adiabatic passage (STIRAP) to $|r_1\rangle=|68D_{5/2},m_J=5/2\rangle$. Then, the incident MW field couples $|r_1\rangle$ to $|r_2\rangle=|69P_{3/2},m_J=3/2\rangle$.

\subsubsection{Single-Atom Homodyne Measurement}

Let $\mathbf E_k(t)$ denote the local MW electric field and $\mathbf d_{12}$ the selected transition-dipole vector. Under resonant and polarization-matched operation, the field-to-Rabi mapping is given by
\begin{equation}
\Omega_{S,k}(t) =\frac{|\mathbf d_{12}\cdot\mathbf E_k(t)|}{\hbar} =\kappa_E|E_k(t)|,
\label{eq:rabi_mapping_sec2}
\end{equation}
where $\kappa_E\triangleq\mu_{\rm dp}/\hbar$ and $\mu_{\rm dp}\triangleq|\mathbf d_{12}|$.

For a sufficiently strong signal, $\Omega_{S,k}$ can be estimated by fitting the observed Rabi oscillation. In the weak-field regime, a calibrated local MW operation prepares a phase reference, after which the atom evolves under the incident field for an interrogation time $\tau_{\rm int}$ \cite{zhang2026microwave}. For an ideal atom, the preparation, signal evolution, and final $\pi/2$ analysis rotation give the ideal single-shot probability of projecting atom $k$ onto state $|r_1\rangle$, which can be expressed as
\begin{equation}
p_{\uparrow,k}^{\rm id} =\frac{1}{2} - \frac{V_k}{2} \cos(\phi_k-\phi_{\rm a}) \sin(\Omega_{S,k}\tau_{\rm int}),
\label{eq:ideal_analysis_probability_sec2}
\end{equation}
where $\phi_k$ is the received-field phase, and $\phi_{\rm a}$ is the phase of the final analysis rotation. $V_k=\exp\!\left(-\frac{\Gamma_{\rm pop}\tau_{\rm int}}{2}\right)\leq1$ collects coherent-visibility loss and $\Gamma_{\rm pop}$ is the population-decay rate.

Let $\epsilon_{0,k}$ and $\epsilon_{1,k}$ denote the $0\!\rightarrow\!1$ and $1\!\rightarrow\!0$ readout-error probabilities at site $k$, respectively. The observed probability is given by
\begin{align}
p_{\uparrow,k} &=\epsilon_{0,k}(1-p_{\uparrow,k}^{\rm id}) + (1-\epsilon_{1,k})p_{\uparrow,k}^{\rm id}\nonumber\\
&=\frac{B_{{\rm sp},k}}{2} - \frac{A_{{\rm eff},k}}{2} \cos(\phi_k-\phi_{\rm a}) \sin(\Omega_{S,k}\tau_{\rm int}),
\label{eq:spam_affine_probability_sec2}
\end{align}
where $B_{{\rm sp},k}=1+\epsilon_{0,k}-\epsilon_{1,k}$, and $A_{{\rm eff},k}=V_k(1-\epsilon_{0,k}-\epsilon_{1,k})$.

After calibration and grouping of atoms with comparable parameters, the site index is suppressed. With $\phi\triangleq\phi_k-\phi_{\rm a}$, the probability becomes \cite{zhang2026microwave}
\begin{equation}
p_{\uparrow} =\frac{B_{\rm sp}}{2} -\frac{A_{\rm eff}}{2}\cos(\phi)\sin(\Omega_S\tau_{\rm int}),
\label{eq:single_atom_probability_sec2}
\end{equation}
where $A_{\rm eff}$ and $B_{\rm sp}$ denote the calibrated state-preparation-and-measurement (SPAM) contrast and offset, respectively.

State mapping and fluorescence discrimination produce a binary outcome $z_k\in\{0,1\}$. For $N$ conditionally independent atoms operated at the same calibrated point, we have\footnote{The binomial model assumes conditionally independent sites with a common calibrated success probability, and unequal site probabilities lead to Poisson--binomial counts. These nonidealities modify the count statistics but not the basic STM scheduling principle.}
\begin{equation}
N_\uparrow=\sum_{k=1}^{N}z_k \sim\operatorname{Bin}(N,p_\uparrow),
\label{eq:single_branch_binomial_sec2}
\end{equation}
where $\operatorname{Bin}(N,p_\uparrow)$ denotes the binomial distribution for $N$ conditionally independent atoms, each of which is detected in the upper state with probability $p_\uparrow$.

In the weak-field region and at maximum slope, the resulting QPN-limited Rabi-frequency uncertainty is given by \cite{zhang2026microwave}
\begin{equation}
\sigma_{\Omega}^{\rm QPN} =\frac{\sqrt{B_{\rm sp}(2-B_{\rm sp})}} {A_{\rm eff}\tau_{\rm int}\sqrt{N}}.
\label{eq:sql_rabi_sec2}
\end{equation}

\subsubsection{Dead Time of Pulsed Measurement Cycle}

The coherent interaction is much faster than the complete destructive measurement cycle \cite{zhang2026microwave}. The complete cycle is $T_{\rm cyc}^{\rm pulse}=\tau_{\rm int}+\tau_{\rm dead}$, and its interrogation duty factor is $\eta_{\rm pulse}=\frac{\tau_{\rm int}}{T_{\rm cyc}^{\rm pulse}}$, where $\tau_{\rm dead}$ denotes the macroscopic dead time that collects all intervals without MW interrogation. Specifically, $\tau_{\rm dead}$ is composed of the times required for atom loading, transport to the interrogation zone, internal-state and Rydberg-state preparation, and state-selective readout, respectively. According to the reported 53-Hz cycle in \cite{zhang2026microwave}, i.e., $T_{\rm cyc}^{\rm pulse}\simeq18.9~\mathrm{ms}$, and $\tau_{\rm int}=20~\mu\mathrm{s}$, we have $\eta_{\rm pulse}\simeq1.06\times10^{-3}$. Although a separate strong-pulse mode exhibits nanosecond-scale microscopic response, a 53-Hz destructive cycle cannot be directly applied to real-time wireless communications.

\begin{figure*}[t]
\centering \includegraphics[width=6.8in]{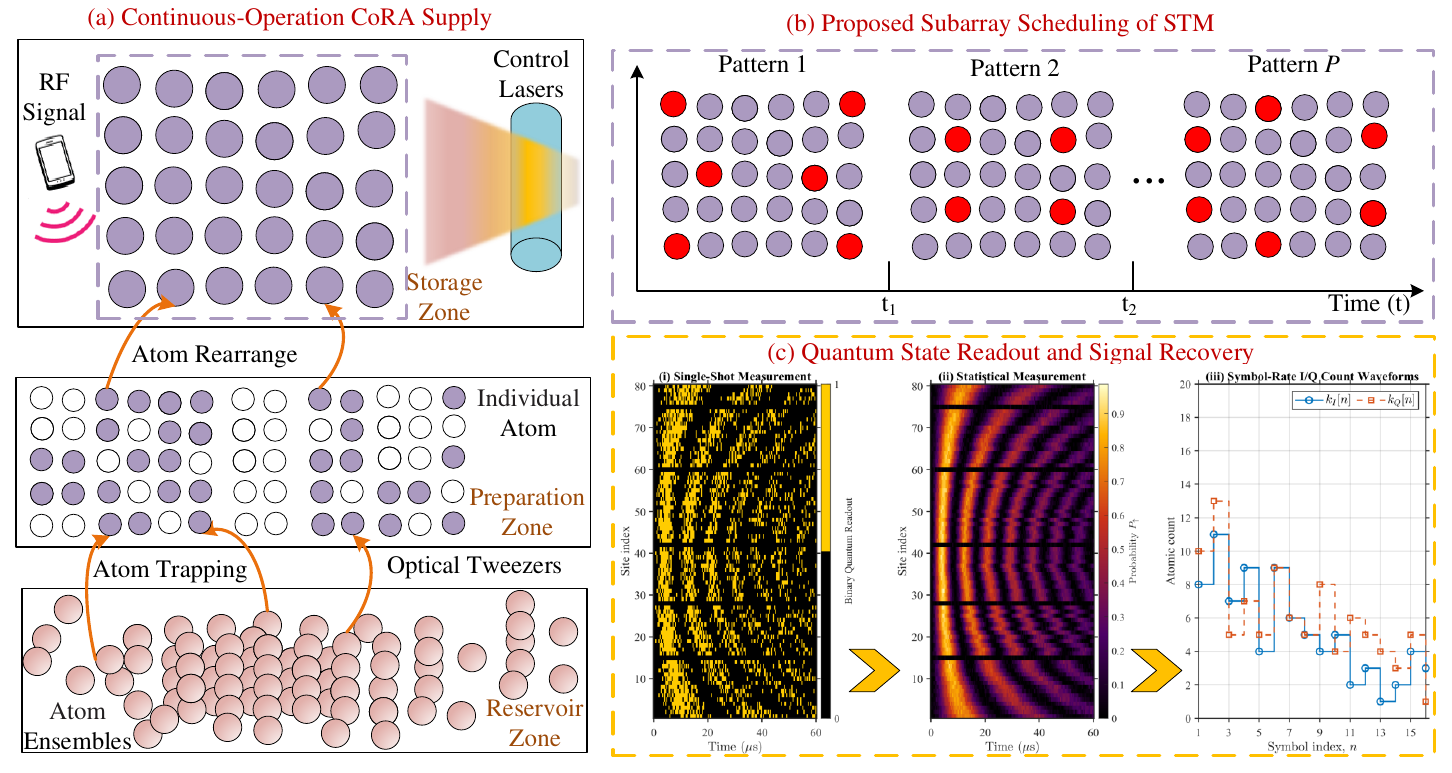}
\caption{Continuous-operation CoRA architecture: (a) Reservoir, preparation, and storage zones for continuous atom supply \cite{chiu2025continuous}, (b) {Proposed temporal subarray scheduling for STM-based wireless reception}, and (c) Quantum state readout and signal recovery.}
\label{fig:continuous_array}
\end{figure*}

\subsection{Continuous-Operation CoRA Architecture}

To remove macroscopic dead time from the foreground sensing cycle, the continuous architecture spatially separates atom generation from coherence-sensitive operations \cite{gyger2024continuous, chiu2025continuous}. As illustrated in Fig.~\ref{fig:continuous_array}(a), atoms progress through three functional zones.

\subsubsection{Reservoir Zone}
The reservoir zone uses a continuous magneto-optical trap (MOT) to capture thermal atoms. A red-detuned one-dimensional optical lattice transports atoms downstream without interrupting the MOT.

\subsubsection{Preparation Zone}
In this zone, multi-channel acousto-optic-deflector (AOD) tweezers extract atoms from the moving lattice, rearrange them into defect-free geometries, and optically pump them into the required initial state.

\subsubsection{Storage Zone}
Initialized arrays are transferred to a static spatial-light-modulator (SLM) tweezer array for MW interrogation and Rydberg excitation.

\subsubsection{Continuous Spectral Shielding Mechanism}
A 1529-nm shielding beam applies spatially selective alternating-current (AC) Stark shifts to the coherence-critical regions, thereby suppressing scattering caused by simultaneous cooling and imaging operations \cite{chiu2025continuous}. Let $\tau_{\rm reset}$ denote the sequential reset time of the pulsed architecture and $\tau_{\rm sw}$ the inter-subarray optical switching time. The pulsed-cycle rate can be expressed as $f_{\rm cyc}^{\rm pulse} =\frac{1}{\tau_{\rm int}+\tau_{\rm reset}} \simeq\frac{1}{\tau_{\rm reset}}$, and the foreground interrogation-rate bound is $f_{\rm fg} \leq\frac{1}{\tau_{\rm int}+\tau_{\rm sw}}$. Hence, ideal cycle-rate gain is given by
\begin{align}
\xi_{\rm cyc} =\frac{f_{\rm fg}}{f_{\rm cyc}^{\rm pulse}} \leq\frac{\tau_{\rm int}+\tau_{\rm reset}} {\tau_{\rm int}+\tau_{\rm sw}} \simeq\frac{\tau_{\rm reset}} {\tau_{\rm int}+\tau_{\rm sw}}.
\label{eq:cycle_gain_sec2}
\end{align}

{\begin{remark}[Continuous-operation hardware vs. communication reception] Although the demonstrated continuous-operation architecture removes the macroscopic atom-loading and preparation stages from the foreground sensing cycle, it does not by itself constitute a wireless communication receiver. Its operation remains based on destructive, site-resolved atomic readout, and the demonstrated hardware does not specify how successive communication symbols should be assigned to fresh atomic ensembles. Therefore, the remaining challenge is not continuous atom supply alone, but a communication-layer co-design that converts the continuously supplied CoRA into a burst wireless receiver.
\end{remark}
}

\section{{Proposed STM-Based} CoRA for Wireless Communications}
\label{sec:short_burst_receiver}

{Building on the demonstrated continuous-operation backbone in Section~\ref{sec:physical_architecture}, this section proposes the STM-based CoRA for burst wireless reception.} Then, we propose the joint binomial particle smoother (JBPS) to perform burst-wise joint state-data inference.

\subsection{Proposed Spatiotemporal Multiplexing Framework}

As illustrated in Fig.~\ref{fig:continuous_array}(b), the proposed STM overlay interleaves atoms in space and schedules the resulting subarrays in time to provide gap-free sampling within an admitted communication burst. To retain both quadratures of the received MW field, each symbol is simultaneously interrogated by two atomic subarrays that form an I/Q pair. The two subarrays experience the same incident MW signal and therefore the same signal-induced Rabi angle, but they are measured using local analysis rotations whose phases differ by $\pi/2$.

\subsubsection{Interleaved Spatial Partition}

For an $R\times C$ tweezer lattice with $K=RC$ sites, the lattice index set is given by
\begin{equation}
\mathcal U =\{(r,c):r=0,\ldots,R-1,\;c=0,\ldots,C-1\}.
\label{eq:lattice_index_sec2}
\end{equation}
The sites are partitioned into $M$ scheduled I/Q pairs or $2M$ disjoint subarrays, which are given by
\begin{align}
q(r,c) &=(r\Delta_r+c\Delta_c)\bmod(2M), \label{eq:subarray_partition_sec2_q}\\
\mathcal A_{I,m} &=\{(r,c)\in\mathcal U:q(r,c)=2m\},\nonumber\\
\mathcal A_{Q,m} &=\{(r,c)\in\mathcal U:q(r,c)=2m+1\},
\label{eq:subarray_partition_sec2}
\end{align}
for $m=0,\ldots,M-1$, where $\Delta_r$ and $\Delta_c$ distribute each group over the lattice footprint. In the balanced case, $|\mathcal A_{I,m}|=|\mathcal A_{Q,m}|=K_m$, so one I/Q snapshot consumes $2K_m$ atoms.

\subsubsection{Temporal Scheduling and Foreground Duty Factor}

\begin{figure}[t]
\centering \includegraphics[width=3.3in]{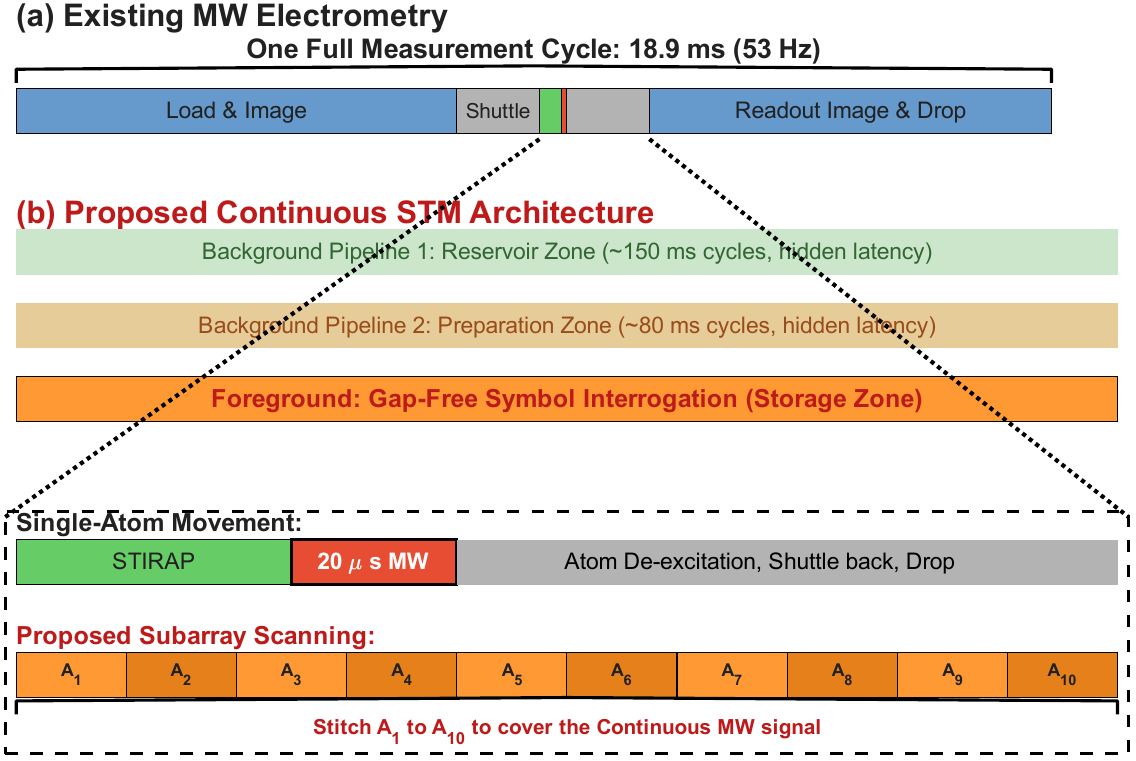}
\caption{Timing comparison {between (a) the existing pulsed MW electrometry \cite{zhang2026microwave} and (b) the proposed STM-based CoRA architecture}, which preloads disjoint I/Q subarray pairs and selects one fresh pair per symbol.}
\label{fig:stm_timing}
\end{figure}

Let $T_{\rm sym}$ denote the symbol interval and $m[n]=n\bmod M$ the I/Q pair assigned to symbol $n$. Only pair $m[n]$ is activated during that interval. The symbol-level timing constraint is given by
\begin{equation}
T_{\rm sym} \geq T_{\rm micro} \triangleq\tau_{\rm int}+\tau_{\rm sw}.
\label{eq:micro_cycle_sec2}
\end{equation}
The aggregate snapshot rate is $f_{\rm samp}=1/T_{\rm sym}$, whereas each pair is activated at $f_{\rm pair}=1/(MT_{\rm sym})$, such that $f_{\rm samp}=Mf_{\rm pair}$. The coherent-interrogation fill factor is given by
\begin{equation}
\eta_{\rm int} =\frac{\tau_{\rm int}}{T_{\rm sym}} \leq\frac{\tau_{\rm int}} {\tau_{\rm int}+\tau_{\rm sw}}.
\label{eq:stm_fill_factor_sec2}
\end{equation}

\subsubsection{Recovery Latency and Atom-Flux Constraint}

After an I/Q pair is interrogated, it has $(M-1)T_{\rm sym}$ for readout, ejection or reinitialization, transport, and refill. Let $\tau_{\rm rec}$ denote this end-to-end recovery time and $\Phi_{\rm q}$ the usable initialized-atom flux. The pair-reuse latency can be expressed as
\begin{align}
(M-1)T_{\rm sym} &\geq\tau_{\rm rec}.
\label{eq:latency_condition_sec2}
\end{align}
For indefinite cyclic STM operation without burst buffering, the initialized-atom supply would additionally need to satisfy
\begin{equation}
\frac{2K_m}{T_{\rm sym}}\le \Phi_q.
\label{eq:flux_condition_sec2}
\end{equation}
This stringent condition motivates the buffered burst protocol introduced in Section~III-B.

\begin{remark}[Temporal bottleneck versus spatial multiplexing]
Note that the proposed STM architecture does not shorten the intrinsic atomic recovery time. Instead, it converts the local temporal bottleneck into a spatial-resource requirement by rotating fresh subarrays across successive symbol intervals. Consequently, sufficient spatial resources can provide gap-free interrogation within an admitted burst, whereas indefinitely sustainable operation remains constrained by the atomic recovery and replenishment rates.
\end{remark}

\subsubsection{Atomic Spacing Constraint}

The relevant blockade quantity is the minimum separation among atoms simultaneously excited in an I/Q pair, which is given by
\begin{equation}
d_{\min}^{(m)}
=\min_{\substack{i\neq j\\
i,j\in\mathcal A_{I,m}\cup\mathcal A_{Q,m}}}
\|\mathbf r_i-\mathbf r_j\|_2.
\label{eq:min_spacing_sec2}
\end{equation}

In general, the blockade-radius $R_b$ depends on the selected Rydberg states and laser parameters \cite{browaeys2020many}. For typical principal quantum numbers, e.g., $n = 50 \sim 70$ utilized in state-of-the-art Rydberg atom arrays, the blockade radius generally spans $R_b \approx 5 \sim 10\ \mu\text{m}$ \cite{bluvstein2024logical}. The partition should therefore satisfy $d_{\min}^{(m)}\gtrsim R_b$ with an appropriate safety margin.

\subsubsection{Observation Mapping}

Let $\mathbf z[n]\in\{0,1\}^{K}$ denote the binary site outcomes over the lattice. For branch $b\in\{I,Q\}$ and pair $m$, let $\mathbf S_{b,m}\in\{0,1\}^{K_m\times K}$ select the sites in $\mathcal A_{b,m}$. The baseband observations are given by
\begin{equation}
\mathbf z_{b,m[n]}[n] =\mathbf S_{b,m[n]}\mathbf z[n], \qquad b\in\{I,Q\},
\label{eq:selection_mapping_sec2}
\end{equation}
together with the known activity index $m[n]$. The corresponding branch-wise atomic counts are $k_b[n] = \mathbf 1^{T}\mathbf z_{b,m[n]}[n],$ and the symbol-rate observation is $\mathbf k_n=[k_I[n],k_Q[n]]^{T}$.

The incident field illuminates both groups, while their analysis rotations differ by $\pi/2$. Hence, the I and Q counts are obtained simultaneously from disjoint ensembles, whose symbol-rate observations $k_I[n]$ and $k_Q[n]$ are illustrated in Fig.~\ref{fig:continuous_array}(c). The corresponding temporal scheduling of fresh I/Q subarray pairs over successive symbol intervals is further shown in Fig.~\ref{fig:stm_timing}, illustrating the timing advantage of the proposed STM architecture. By moving the slow atom-loading and preparation stages to background pipelines and sequentially activating preloaded subarrays, the proposed STM converts the macroscopic dead-time bottleneck into a spatial-resource requirement and enables gap-free symbol interrogation within an admitted burst.

\subsection{Communication Protocol and Atomic Observation}
\label{subsec:communication_atomic_observation}

Since each I/Q snapshot destructively measures two groups of $K_m$ atoms, indefinite symbol-rate reception would require an initialized-atom flux $2K_m/T_{\rm sym}$, which can greatly exceed the background flux of continuous-operation cold-atom arrays. The proposed quantum-constrained burst protocol therefore treats the storage array as a preloaded atomic buffer rather than claiming indefinite fresh-atom sampling. The burst-admission constraint for $N_{\rm b}\le M$ I/Q snapshots is $2K_mN_{\rm b}\leq K_0$, where $K_0$ is the number of initialized atoms available at the burst start. Across repeated bursts with usable initialized-atom flux $\Phi_{\rm q}$, the average atom-supply constraint is $2K_mN_{\rm b}f_{\rm pkt}\leq\Phi_{\rm q}$, where $f_{\rm pkt}$ is the long-term packet repetition frequency (PRF).

Let $\mathcal P$ and $\mathcal D$ denote the disjoint pilot and payload index sets of an $N_{\rm b}$-symbol burst, with $\mathcal P\cup\mathcal D=\{0,\ldots,N_{\rm b}-1\}$. A payload symbol $x[n]$ is drawn from a finite alphabet $\mathcal C$ with prior $P_X(c)$ and $\mathbb E[|x[n]|^2]=1$. Let $M_b=|\mathcal C|=2^{m_b}$ and $\mathbf b(c)=[b_1(c),\ldots,b_{m_b}(c)]$ denote its binary label. The transmitted complex baseband burst is given by
\begin{equation}
u(t)=\sum_{n=0}^{N_{\rm b}-1}x[n]g_{\rm c}(t-nT_{\rm sym}),
\label{eq:cosine_pulse_shaped_burst}
\end{equation}
where $g_{\rm c}$ denotes the one-symbol cosine pulse.

Given transmit power $P_t$, transmitter gain $G_t$, and link distance $d$, the free-space power density and peak electric-field amplitude are given by
\begin{align}
E_0(P_t,d)&=\sqrt{2Z_0\frac{P_tG_t}{4\pi d^2}},
\label{eq:free_space_power_to_field_r5}
\end{align}
where $Z_0$ is the free-space impedance.

The small-scale fading and packet-to-packet shadowing can be modeled as
\begin{equation}
h_{\rm ch}=10^{\xi/20}\!\left[ \sqrt{\frac{K_R}{K_R+1}}e^{j\phi_{\rm L}} +\sqrt{\frac{1}{K_R+1}}g \right],
\label{eq:block_rician_channel_r5}
\end{equation}
where $g\sim\mathcal{CN}(0,1)$, $K_R$ is the linear Rician factor, $\phi_{\rm L}$ is the line-of-sight phase, and $\xi\sim\mathcal N(0,\sigma_{\rm sh}^2)$ is the shadowing in dB. The coefficient is constant within a burst and independently regenerated between bursts.

After timing synchronization removes the common delay $d/c_0$, the received peak-field envelope is $\mathcal E_{\rm rx}(t)=E_0(P_t,d)h_{\rm ch}u(t-d/c_0)$. The output of the rectangular atomic gate with duration $\tau_{\rm int}$ can be expressed as \cite{zhang2026microwave}
\begin{align}
\overline{\mathcal E}_n =\frac{1}{\tau_{\rm int}} \int_{nT_{\rm sym}}^{nT_{\rm sym}+\tau_{\rm int}} \mathcal E_{\rm rx}(t)\,dt =\alpha_gE_0(P_t,d)h_{\rm ch}x[n],
\label{eq:atomic_integrate_dump}
\end{align}
where $\alpha_g\triangleq\frac{1}{\tau_{\rm int}}\int_0^{\tau_{\rm int}}g_{\rm c}(t)\,dt$.

Let $\zeta_{\rm pol}\in[0,1]$ be the projection of the incident polarization onto the selected transition dipole \cite{sedlacek2013vector}. The signal Rabi phasor before residual reference drift is given by
\begin{equation}
\Omega_n =\zeta_{\rm pol}\kappa_E\overline{\mathcal E}_n =\alpha_g\zeta_{\rm pol}\kappa_EE_0(P_t,d)h_{\rm ch}x[n],
\label{eq:wireless_rabi_phasor_r6}
\end{equation}
where $\kappa_E=\mu_{\rm dp}/\hbar$ and $\mu_{\rm dp}$ is the transition dipole moment.

\subsection{Atom State Tracking and Binomial Likelihood}
\label{subsec:dynamic_binomial_likelihood}

\subsubsection{Dynamic Phase and Rabi-Gain State}

Coherent symbol recovery depends on the phase of the incident Rabi phasor relative to the phase-coherent local MW analysis operations \cite{Simons2019RydbergMixer}. Although the block-fading coefficient $h_{\rm ch}$ remains constant within one burst, its amplitude and phase are unknown at the receiver. Moreover, a residual frequency offset between the received carrier and the analysis reference accumulates phase from symbol to symbol, while their relative oscillator phase noise introduces random phase innovations. Accordingly, the receiver must infer both the relative phase and the effective field-to-Rabi gain from the atomic observations.

Let $f_{\rm rx}(t)$ and $f_{\rm ref}(t)$ denote the instantaneous frequencies of the received carrier and the local MW analysis reference, respectively, and define $\Delta f(t)\triangleq f_{\rm rx}(t)-f_{\rm ref}(t)$. The continuous relative phase can then be expressed as
\begin{equation}
\theta(t) = \theta_0 +2\pi\int_0^t\Delta f(\tau)\,d\tau +\phi_{\rm pn}(t),
\label{eq:continuous_relative_phase}
\end{equation}
where $\theta_0$ is the unknown packet-start composite phase, and $\phi_{\rm pn}(t)$ represents the residual relative phase noise of the received carrier and the linewidth of lasers.

Let $t_n=nT_{\rm sym}$. The phase accumulated over the interval $[t_n,t_{n+1})$ is defined as
\begin{equation}
\nu_n \triangleq 2\pi\int_{t_n}^{t_{n+1}}\Delta f(\tau)\,d\tau \simeq 2\pi\Delta f_nT_{\rm sym},
\label{eq:phase_increment_definition}
\end{equation}
where $\nu_n$ is measured in radians per symbol and the approximation holds when $\Delta f(t)$ varies slowly within one symbol interval. The corresponding phase-noise innovation is given by
\begin{equation}
w_{\theta,n} \triangleq \phi_{\rm pn}(t_n)-\phi_{\rm pn}(t_{n-1}).
\label{eq:phase_noise_innovation}
\end{equation}

To describe the unknown amplitude scale, let $ g_n=g_{\rm ref}e^{\ell_n}$ denote the effective signal Rabi-frequency amplitude produced by a unit-magnitude symbol during the $n$th interrogation, where $g_{\rm ref}>0$ is a nominal calibrated Rabi-frequency gain. Its normalized logarithmic representation is $ \ell_n \triangleq \ln\!\left(\frac{g_n}{g_{\rm ref}}\right)$, and $\Omega_n=g_n e^{j\theta_n}c$ for candidate symbol $c$. Spatial MW-field inhomogeneity and residual subarray-to-subarray calibration errors can produce small variations in the effective gain $ \ell_n$\cite{manetsch2025tweezer}. Sampling the relative phase once per symbol gives the componentwise state dynamics as
\begin{align}
\theta_n &= \theta_{n-1}+\nu_{n-1}+w_{\theta,n}, \nonumber\\
\nu_n &= \rho_\nu\nu_{n-1}+w_{\nu,n}, \nonumber\\
\ell_n &= \ell_{n-1}+w_{\ell,n},
\label{eq:component_state_recursions}
\end{align}
where $w_{\theta,n}$, $w_{\nu,n}$, and $w_{\ell,n}$ denote the phase, phase-increment, and logarithmic-gain innovations, respectively. The first recursion follows directly from the phase accumulated over one symbol interval. The second recursion models the slowly varying residual frequency offset as a first-order Gauss-Markov process \cite{aghamohammadi1989adaptive}, where $0\leq\rho_\nu\leq1$ is its symbol-to-symbol correlation coefficient. In particular, $\rho_\nu=1$ gives a random-walk frequency offset. The third recursion assumes that the effective Rabi gain is approximately constant within a short burst, while allowing small residual amplitude and calibration variations.

Hence, the resulting phase-frequency-gain dynamics can be represented in a linear Gaussian state-space \cite{aghamohammadi1989adaptive}. Specifically, the discrete state evolution can be expressed as
\begin{align}
&\mathbf s_n=\mathbf F\mathbf s_{n-1}+\mathbf w_n,
\nonumber \\
&\mathbf F=
\begin{bmatrix}
1&1&0\\
0&\rho_\nu&0\\
0&0&1
\end{bmatrix},
\label{eq:jbps_state_transition_r6}
\end{align}
where $\mathbf s_n= [\theta_n,\,\nu_n,\,\ell_n]^{\mathsf T}$, $\mathbf w_n= [w_{\theta,n},\,w_{\nu,n},\,w_{\ell,n}]^{\mathsf T} \sim\mathcal N(\mathbf 0,\mathbf Q)$, and $\mathbf Q= \operatorname{diag}(q_\theta,q_\nu,q_\ell)$. Here, $q_\theta$, $q_\nu$, and $q_\ell$ are the corresponding process-noise variances.

The initial state is given by
\begin{align}
\theta_0&=\arg(h_{\rm ch}),\nonumber\\
\nu_0&\simeq2\pi\Delta f_0T_{\rm sym},\nonumber\\
\ell_0&=\ln\!\left[ \frac{\alpha_g\zeta_{\rm pol}\kappa_EE_0(P_t,d)|h_{\rm ch}|} {g_{\rm ref}} \right]+\delta_{\ell,0},
\label{eq:composite_channel_state_initialization_r6}
\end{align}
where $\delta_{\ell,0}$ captures residual gain-calibration uncertainty.

For candidate symbol $c\in\mathcal C$, the corresponding phase and Rabi angle are defined as $\psi_n(c)=\theta_n+\arg(c)$ and $q_n(c)=g_n|c|\tau_{\rm int}$, respectively. The conditional excitation probabilities of the two disjoint analysis branches are given by
\begin{align}
p_{I,n}(c,\mathbf s_n) &=\frac{B_{\rm sp}}{2} -\frac{A_{\rm eff}}{2}\cos[\psi_n(c)]\sin[q_n(c)],\nonumber\\
p_{Q,n}(c,\mathbf s_n) &=\frac{B_{\rm sp}}{2} -\frac{A_{\rm eff}}{2}\sin[\psi_n(c)]\sin[q_n(c)].
\label{eq:jbps_branch_probabilities_r6}
\end{align}

The phase enters through the angular factors, whereas the symbol magnitude and Rabi gain enter through the periodic factor $\sin q_n$. This periodicity makes local response inversion noise-amplifying and potentially ambiguous near and beyond a Rabi fold. With $K_m$ atoms in each branch, the observed excited-state count is distributed as
\begin{equation}
k_b[n]\mid c,\mathbf s_n \sim\operatorname{Bin}\!\left(K_m,p_{b,n}(c,\mathbf s_n)\right), \quad b\in\{I,Q\}.
\label{eq:jbps_binomial_counts_r6}
\end{equation}

Since the I and Q groups are disjoint, their outcomes are conditionally independent under the main model. Hence, for $\mathbf k_n=[k_I[n],k_Q[n]]^{\mathsf T}$, the joint likelihood is given by
\begin{equation}
p(\mathbf k_n\mid c,\mathbf s_n) =\prod_{b\in\{I,Q\}} \binom{K_m}{k_b[n]} p_{b,n}^{k_b[n]}(1-p_{b,n})^{K_m-k_b[n]}.
\label{eq:full_binomial_likelihood_sec3}
\end{equation}
When hypotheses have the same $K_m$ and observed counts, the common binomial coefficients may be omitted. The resulting exact log-metric can be expressed as
\begin{align}
\Lambda_n(c,\mathbf s_n) &=\sum_{b\in\{I,Q\}}\bigl\{ k_b[n]\ln p_{b,n}(c,\mathbf s_n)\nonumber\\
&+[K_m-k_b[n]]\ln[1-p_{b,n}(c,\mathbf s_n)] \bigr\}.
\label{eq:jbps_log_likelihood_r6}
\end{align}

The conventional plug-in detector inserts one predicted state into this metric and marginalizes it over candidates with the desired bit label. Typically, a Kalman filter is first employed to obtain a point estimate of the time-varying channel state \cite{aghamohammadi1989adaptive}, after which this estimate is treated as the true state for subsequent maximum-likelihood data detection. Such point-estimate-then-detect processing neglects residual state uncertainty and possible multimodality, motivating the proposed singal detection scheme in Section \ref{subsec:jbps_receiver_r6}.

\subsection{Proposed Joint Binomial Particle Smoother-Based Signal Detection Scheme for CoRA Systems}
\label{subsec:jbps_receiver_r6}
The proposed JBPS combines particle filtering for nonlinear and non-Gaussian state inference with forward-filtering backward-simulation smoothing\cite{978374, godsill2004montecarlo}, while retaining the exact binomial likelihood and marginalizing unknown payload symbols during state tracking. Since the nonlinear and periodic binomial observation model can produce a non-Gaussian and potentially multimodal state posterior, JBPS represents the state uncertainty using a finite set of weighted particles. A particle is a numerical sample representing one plausible hypothesis for the latent receiver state. The forward particle filter represents the generally non-Gaussian state posterior by $N_{\rm PF}$ weighted hypotheses $\{\mathbf s_n^{(i)},w_n^{(i)}\}_{i=1}^{N_{\rm PF}}$. For ancestor index $a_i$, the predicted particle is generated according to $\mathbf s_n^{(i)} \sim p(\mathbf s_n\mid\mathbf s_{n-1}^{(a_i)}),$ using covariance $\mathbf Q$. The propagated phase is wrapped to $(-\pi,\pi]$, whereas $\nu_n$ and $\ell_n$ remain real-valued. For particle $i$, the observation likelihood is given by
\begin{equation}
\mathcal L_n^{(i)}=\begin{cases}
p(\mathbf k_n\mid x_{\rm p}[n],\mathbf s_n^{(i)}), &n\in\mathcal P,\\[1mm]
\displaystyle\sum_{c\in\mathcal C}P_X(c) p(\mathbf k_n\mid c,\mathbf s_n^{(i)}), &n\in\mathcal D.
\end{cases}
\label{eq:jbps_measurement_factor_r6}
\end{equation}
Thus, an uncertain payload observation updates the state through a soft symbol mixture rather than a hard virtual pilot. With the state transition used as the proposal distribution, the filtering weights are given by
\begin{equation}
\widetilde w_n^{(i)} =\overline w_{n-1}^{(i)}\mathcal L_n^{(i)}, \qquad w_n^{(i)} =\frac{\widetilde w_n^{(i)}} {\sum_{q=1}^{N_{\rm PF}}\widetilde w_n^{(q)}}.
\label{eq:jbps_forward_weight_r6}
\end{equation}
Resampling is triggered when $N_{\rm eff}^{\rm PF}[n] =[\sum_i(w_n^{(i)})^2]^{-1} <\alpha_{\rm res}N_{\rm PF}$, after which the inherited weights are reset to $1/N_{\rm PF}$.

After the forward pass, forward-filtering backward-simulation draws $N_{\rm s}$ complete state trajectories \cite{godsill2004montecarlo}. If particle $i_{n+1}^{(r)}$ has been selected at time $n+1$, the normalized probability of selecting predecessor $i$ is given by
\begin{equation}
\beta_n^{(i,r)} = \frac{ w_n^{(i)} p\!\left( \mathbf s_{n+1}^{(i_{n+1}^{(r)})} \mid\mathbf s_n^{(i)} \right)} {\displaystyle \sum_{q=1}^{N_{\rm PF}} w_n^{(q)} p\!\left( \mathbf s_{n+1}^{(i_{n+1}^{(r)})} \mid\mathbf s_n^{(q)} \right)}.
\label{eq:ffbsi_backward_weight_r6}
\end{equation}

Conditioned on the $r$th smoothed trajectory, the candidate-symbol posterior is given by
\begin{equation}
\pi_n^{(r)}(c) = \frac{ P_X(c)\exp[\Lambda_n(c,\mathbf s_n^{(r)})]} {\displaystyle \sum_{c'\in\mathcal C} P_X(c')\exp[\Lambda_n(c',\mathbf s_n^{(r)})]}.
\label{eq:trajectory_symbol_posterior_r6}
\end{equation}

The smoothed bit posterior probability can be expressed as
\begin{align}
\widehat P_{j,n}(a) =\frac{1}{N_{\rm s}} \sum_{r=1}^{N_{\rm s}} \sum_{\substack{c\in\mathcal C:\\b_j(c)=a}} \pi_n^{(r)}(c).
\end{align}

Then, the final decision is given by
\begin{align}
\widehat b_j[n] =\arg\max_{a\in\{0,1\}} \widehat P_{j,n}(a).
\label{eq:jbps_bit_decision_r6}
\end{align}

\section{Performance Analysis for CoRA Systems}
\label{sec:performance_analysis_r6}

In this section, we derive and compare the sensitivity between the classical antenna, vapor-cell receiver, and CoRA on a common input-field and noise-bandwidth reference plane. Then, we provide the finite-alphabet rate and atom-constrained PRF for CoRA assisted wireless communication systems.

\subsection{SNR Analysis for Classical Antenna, Vapor-Cell, and Atom-Array Receivers}
\label{subsec:unified_snr_benchmark_r6}

\subsubsection{Common Reference Plane and Noise Convention}

All receivers are compared at the same carrier frequency and one-sided equivalent noise bandwidth $B_{\rm n}$. The common incident peak-field amplitude is given by
\begin{equation}
E=E_0(P_t,d)|h_{\rm ch}|.
\label{eq:common_incident_field_r6}
\end{equation}

In the locally linear weak-signal region, the input-referred noise-equivalent field (NEF) $\mathcal E_r$ denotes the one-sided amplitude spectral density of the equivalent input field noise \cite{bussey2022quantum}. Assuming that this noise is approximately white over the equivalent noise bandwidth $B_{\rm n}$, the weak-signal SNR is given by
\begin{equation}
\gamma_r^{\rm ws}(E) = \frac{E^2}{\mathcal E_r^2B_{\rm n}}.
\label{eq:unified_nef_snr_r6}
\end{equation}

\subsubsection{Classical Antenna-Based Receiver}

For a copolarized receiving antenna with realized gain $G_r$ and wavelength $\lambda$, the effective aperture is given by $A_{\rm e}=\frac{G_r\lambda^2}{4\pi}$. The available signal and thermal-noise powers are given by \cite{Balanis2016Antenna}
\begin{equation}
P_{\rm r}=\frac{E^2A_{\rm e}}{2Z_0}, \qquad P_{\rm n}=k_{\rm B}T_{\rm sys}B_{\rm n},
\label{eq:classical_signal_noise_power_r6}
\end{equation}
where $T_{\rm sys}=T_{\rm ant}+(10^{\mathrm{NF}/10}-1)T_0$.

Hence, the SNR and input-referred NEF are given by
\begin{align}
&\gamma_{\rm ant}(E) =\frac{E^2A_{\rm e}} {2Z_0k_{\rm B}T_{\rm sys}B_{\rm n}}, \label{eq:classical_antenna_snr_r6}\\
&\mathcal E_{\rm ant} =\sqrt{\frac{2Z_0k_{\rm B}T_{\rm sys}}{A_{\rm e}}}.
\label{eq:classical_antenna_nef_r6}
\end{align}

\subsubsection{Vapor Cell-Based Atomic Receiver}

Following the RF-superheterodyne and balanced coherent optical detection model in \cite{ref_gong_overview}, the total RF Rabi frequency around the strong local-oscillator operating point is given by
\begin{equation}
\Omega_{{\rm RF},r}(t) = \Omega_{{\rm LO},r} +\Omega_{\rm s}(t), \qquad \Omega_{\rm s}(t) =\frac{\mu_{\rm RF}}{\hbar}E(t),
\label{eq:vapor_rf_rabi_model}
\end{equation}
where $E(t)$ is the weak received RF electric-field component, $\mu_{\rm RF}$ is the dipole moment of the selected RF transition, and $\Omega_{{\rm LO},r}$ is the Rabi frequency induced by the strong RF local oscillator. The weak-signal approximation requires $|\Omega_{\rm s}|\ll|\Omega_{{\rm LO},r}|$.

At this RF-LO-biased operating point, the atomic transconductance and transmitted probe power are given by
\begin{align}
&\chi_{{\rm s},r}^{\prime} \triangleq \left. \frac{\partial\chi_r}{\partial\Omega_{\rm RF}} \right|_{\Omega_{\rm RF}=\Omega_{{\rm LO},r}}, \label{eq:vapor_atomic_transconductance}\\
&P_{{\rm out},r} = P_{\rm p} \exp\!\left[ -\frac{2\pi L_{\rm cell}}{\lambda_{\rm p}} \operatorname{Im} \left\{\chi_r(\Omega_{{\rm LO},r})\right\} \right],
\label{eq:vapor_probe_output_power}
\end{align}
where $\chi_r$ is the Doppler-averaged probe susceptibility, $P_{\rm p}$ is the incident probe power, $L_{\rm cell}$ is the cell length, and $\lambda_{\rm p}$ is the probe wavelength. The corresponding field-to-optical conversion coefficient is given by
\begin{equation}
\kappa_{1,r} = \frac{\pi L_{\rm cell}\mu_{\rm RF}} {\lambda_{\rm p}\hbar} \left|\chi_{{\rm s},r}^{\prime}\right|.
\label{eq:vapor_kappa_one}
\end{equation}

After balanced coherent optical detection and the subsequent photodetector amplification, the effective power gain is given by
\begin{equation}
\beta_r = 4\alpha_{\rm pd}^{2}G_{\rm PD}Z_0 P_{\rm LO}^{\rm opt}P_{{\rm out},r} \kappa_{1,r}^{2}\sin^{2}(\phi_{1,r}),
\label{eq:vapor_effective_power_gain}
\end{equation}
where $\alpha_{\rm pd}$ is the photodetector responsivity, $G_{\rm PD}$ is the photodetection-chain power gain, $P_{\rm LO}^{\rm opt}$ is the optical local-oscillator power, and $Z_0$ is the free-space impedance. Moreover, $\phi_{1,r}$ is the total atomic--optical phase difference, which can be expressed as
\begin{equation}
\phi_{1,r} = \phi_{\rm LO}^{\rm opt} -\phi_{{\rm p},r}(\Omega_{{\rm LO},r}) -\arg\!\left(\chi_{{\rm s},r}^{\prime}\right).
\label{eq:vapor_atomic_optical_phase}
\end{equation}
The optical homodyne phase is locked to the maximum-slope quadrature, such that $|\sin(\phi_{1,r})|\simeq1$.

For a received peak electric-field amplitude $E$, the corresponding output signal power is given by
\begin{equation}
P_{{\rm S},r}(E) = \frac{\beta_rE^2}{2Z_0}.
\label{eq:vapor_signal_power_field}
\end{equation}

Consistent with \cite{ref_gong_overview}, the adopted vapor-cell noise power contains intrinsic electronic thermal noise (ITN), photon shot noise (PSN), and atomic quantum projection noise (QPN), which are given by
\begin{align}
N_{{\rm ITN},r} &= k_{\rm B}T_{\rm PD}G_{\rm PD}B_{\rm n}, \label{eq:vapor_itn_power}\\
N_{{\rm PSN},r} &= 2q_{\rm e}G_{\rm PD}\alpha_{\rm pd} \left(P_{\rm LO}^{\rm opt}+P_{{\rm out},r}\right)B_{\rm n}, \label{eq:vapor_psn_power}\\
N_{{\rm QPN},r} &= \frac{\beta_r}{2Z_0} \left( \frac{\hbar\sqrt{\Gamma_{{\rm coh},r}}} {\mu_{\rm RF}\sqrt{N_{{\rm at},r}^{\rm eff}}} \right)^{\!2} B_{\rm n},
\label{eq:vapor_qpn_power}
\end{align}
where $T_{\rm PD}$ is the equivalent electronic-noise temperature, $q_{\rm e}$ is the elementary charge, $B_{\rm n}$ is the post-detection noise bandwidth, $\Gamma_{{\rm coh},r}$ is the effective atomic-coherence decay rate used in the QPN model, and $N_{{\rm at},r}^{\rm eff}$ is the number of coherently participating atoms.

Therefore, the weak-signal SNR of vapor-cell receiver $r$ is given by
\begin{equation}
\gamma_{{\rm vap},r}^{\rm ws}(E) = \frac{P_{{\rm S},r}(E)} {N_{{\rm ITN},r}+N_{{\rm PSN},r}+N_{{\rm QPN},r}} = \frac{\beta_rE^2} {2Z_0N_{{\rm tot},r}},
\label{eq:vapor_weak_signal_snr}
\end{equation}
where $N_{{\rm tot},r}\triangleq N_{{\rm ITN},r}+N_{{\rm PSN},r}+N_{{\rm QPN},r}$. Under the definition $\gamma_{{\rm vap},r}^{\rm ws}(E) =E^2/(\mathcal E_{{\rm vap},r}^{2}B_{\rm n})$, the input-referred noise-equivalent field is given by
\begin{equation}
\mathcal E_{{\rm vap},r} = \sqrt{ \frac{2Z_0N_{{\rm tot},r}} {\beta_rB_{\rm n}} }.
\label{eq:vapor_nef_code_aligned}
\end{equation}

Note that Eqs. \eqref{eq:vapor_atomic_transconductance}--\eqref{eq:vapor_nef_code_aligned} apply only within the locally linear RF-LO-biased regime. The strong-field results are instead calculated from the full nonlinear susceptibility $\chi_r(\Omega_{\rm RF})$ \cite{11692870}.

\subsubsection{CoRA-Based Atomic Receiver}

At the common incident-field plane, the effective field-to-Rabi coefficient is defined as
\begin{equation}
\kappa_{\rm a}\triangleq\alpha_g\zeta_{\rm pol}\kappa_E.
\label{eq:effective_array_field_rabi}
\end{equation}
Conditioned on tracked phase $\psi$, the field-amplitude Fisher information of the exact binomial observation is given by \cite{sengijpta1995fundamentals}
\begin{equation}
J_E(E,\psi) =K_m\sum_{b\in\{I,Q\}} \frac{[\partial p_b(E,\psi)/\partial E]^2} {p_b(E,\psi)[1-p_b(E,\psi)]}.
\label{eq:array_field_fisher}
\end{equation}

\begin{lemma}[Weak-field array limit]
\label{lem:array_weak_field_nef}
For $q=\kappa_{\rm a}E\tau_{\rm int}\ll1$ and conditioned on $\psi$, the weak-field Fisher information is given by
\begin{equation}
J_E^{\rm ws} =\frac{K_m(A_{\rm eff}\kappa_{\rm a}\tau_{\rm int})^2} {B_{\rm sp}(2-B_{\rm sp})}.
\label{eq:lemma1_fisher}
\end{equation}
Consequently, the weak-signal SNR and QPN-limited NEF are given by \cite{zhang2026microwave}
\begin{align}
&\gamma_{\rm arr}^{\rm ws}(E) =E^2J_E^{\rm ws}, \label{eq:lemma1_snr}\\
&\mathcal E_{\rm arr}^{\rm QPN} =\frac{\sqrt{B_{\rm sp}(2-B_{\rm sp})}} {A_{\rm eff}\kappa_{\rm a}\tau_{\rm int}\sqrt{K_mB_{\rm n}}}.
\label{eq:lemma1_nef}
\end{align}
\end{lemma}

\begin{IEEEproof}
See Appendix~\ref{app:proof_lemma1}.
\end{IEEEproof}

\begin{remark}[Phase-robust weak-field sensitivity]
The orthogonal I/Q branches contribute complementary $\cos^2\psi$ and $\sin^2\psi$ terms to the weak-field Fisher information. Since $\cos^2\psi+\sin^2\psi=1$, their joint observation eliminates phase-dependent sensitivity nulls and yields a phase-invariant local amplitude sensitivity under the matched-branch model. However, this phase robustness does not eliminate the need for phase tracking, because \(\psi\) still determines the I/Q symbol mapping in Section \ref{subsec:dynamic_binomial_likelihood}.
\end{remark}

\subsubsection{Conditional Cross-Receiver Sensitivity Comparison}

The total array NEF, including QPN, technical array noise, and external field noise, is given by
\begin{equation}
(\mathcal E_{\rm arr}^{\rm tot})^2 =(\mathcal E_{\rm arr}^{\rm QPN})^2 +(\mathcal E_{\rm arr}^{\rm tech})^2 +\mathcal E_{\rm ext}^2.
\label{eq:total_array_nef}
\end{equation}

\begin{lemma}[Receiver sensitivity crossover]
\label{lem:receiver_crossover}
At the same input-field reference plane and $B_{\rm n}$, the atom array has higher local weak-signal SNR than receiver $r$ if and only if the following condition is satisfied:
\begin{equation}
\mathcal E_{\rm arr}^{\rm tot}<\mathcal E_r.
\label{eq:lemma2_nef_condition}
\end{equation}
Under the QPN-only array model, the crossover atom numbers for the antenna and vapor-cell receivers are given by
\begin{align}
&K_{m,\rm ant}^{\star} =\frac{B_{\rm sp}(2-B_{\rm sp})A_{\rm e}} {2Z_0k_{\rm B}T_{\rm sys}B_{\rm n} (A_{\rm eff}\kappa_{\rm a}\tau_{\rm int})^2}, \label{eq:lemma2_ant_threshold}\\
&K_{m,\rm vap}^{\star}= \frac{B_{\rm sp}(2-B_{\rm sp})\beta_r} {2Z_0N_{{\rm tot},r} \left(A_{\rm eff}\kappa_{\rm a}\tau_{\rm int}\right)^2}.
\label{eq:lemma2_vap_threshold}
\end{align}
\end{lemma}

\begin{IEEEproof}
For a nonzero input field $E$ and noise bandwidth $B_{\rm n}$, the weak-signal SNR is inversely proportional to the squared NEF. Hence, $\gamma_{\rm arr}^{\rm ws}(E)>\gamma_r^{\rm ws}(E)$ holds if and only if $\mathcal E_{\rm arr}^{\rm tot}<\mathcal E_r$. Substituting $(\mathcal E_{\rm arr}^{\rm tot})^2 =C_{\rm arr}/K_m+\mathcal E_{\rm floor}^2$ shows that a finite crossover exists only when $\mathcal E_r^2>\mathcal E_{\rm floor}^2$. Rearranging this inequality gives $K_m>C_{\rm arr}/ (\mathcal E_r^2-\mathcal E_{\rm floor}^2)$. The QPN-only antenna and vapor-cell thresholds follow by setting $\mathcal E_{\rm floor}=0$ and substituting their corresponding NEFs.
\end{IEEEproof}

\subsection{State-Averaged Finite-Alphabet Rate}
\label{subsec:finite_alphabet_rate}

Since the CoRA produces finite-count binomial observations, the Gaussian-input Shannon expression $\log_2(1+\gamma)$ does not characterize the achievable rate of the finite-alphabet atomic count channel. Instead, using the transition law in Eq.~\eqref{eq:full_binomial_likelihood_sec3}, we define the state-conditioned finite-alphabet rate for a fixed state realization $\mathcal S=s$ as
\begin{align}
&R_{\rm FA}(s;P_X) \triangleq \frac{1}{T_{\rm sym}} I(X;\mathbf K\mid\mathcal S=s) \nonumber\\
&= \frac{1}{T_{\rm sym}} \sum_{c\in\mathcal C}P_X(c) \sum_{\mathbf k}P(\mathbf k\mid c,s) \log_2 \frac{P(\mathbf k\mid c,s)} {P(\mathbf k\mid s)},
\label{eq:count_finite_alphabet_rate_sec4}
\end{align}
where $\mathbf K=(K_I,K_Q)$ is the random count pair and
\begin{equation}
P(\mathbf k\mid s) = \sum_{c\in\mathcal C} P_X(c)P(\mathbf k\mid c,s).
\end{equation}
Accordingly, we have
\begin{equation}
0 \leq R_{\rm FA}(s;P_X) \leq \frac{H(X)}{T_{\rm sym}} \leq \frac{\log_2 M_b}{T_{\rm sym}}.
\end{equation}

Given a specified input distribution and state process, the state-averaged finite-alphabet rate is defined as
\begin{equation}
\overline R_{\rm FA} \triangleq \mathbb E_{\mathcal S} \left[ R_{\rm FA}(\mathcal S;P_X) \right].
\label{eq:state_averaged_finite_alphabet_rate_sec4}
\end{equation}

\subsection{Hardware-Constrained Information-PRF Tradeoff}
\label{subsec:hardware_goodput_prf_r6}

Considering one burst consumes $K_{\rm burst}=2K_mN_{\rm b}$ atoms, and the corresponding buffer-feasibility constraint is $K_{\rm burst}\leq K_0.$ The preloaded buffer supports the instantaneous in-burst interrogation rate, whereas nonoverlapping packet scheduling and atom replenishment impose separate long-term ceilings. Specifically, $f_{\rm int}^{(\rm burst)}=1/{T_{\rm sym}}$ is the instantaneous symbol-interrogation rate within an admitted burst. The maximum packet repetition rate permitted by nonoverlapping packet scheduling can be expressed as
\begin{align}
f_{\rm clk}=\frac{1}{N_{\rm b}T_{\rm sym}}.
\label{eq:packet_clock_ceiling}
\end{align}
Furthermore, the maximum long-term packet rate sustained by the available initialized-atom flux, which is given by
\begin{align}
f_{\rm atom}=\frac{\Phi_{\rm q}}{2K_mN_{\rm b}}.
\label{eq:atom_supply_packet_rate}
\end{align}

Consequently, the sustainable physical PRF is given by
\begin{equation}
f_{\rm pkt} =\min\{f_{\rm clk},f_{\rm atom}\}
\label{eq:packet_prf_r6}
\end{equation}

{\begin{proposition}[Sensitivity--sustainability scaling law]
\label{prop:sensitivity_sustainability}
In the weak-field QPN-limited regime and the atom-flux-limited regime $f_{\rm pkt}=f_{\rm atom}$, with all parameters other than $K_m$ fixed, we have $\gamma_{\rm arr}^{\rm ws}\propto K_m, \mathcal E_{\rm arr}^{\rm QPN}\propto K_m^{-1/2},$ and $f_{\rm pkt}\propto K_m^{-1}.$ Moreover, multiplying Eq.~\eqref{eq:lemma1_snr} by Eq.~\eqref{eq:atom_supply_packet_rate} gives
\begin{equation}
\gamma_{\rm arr}^{\rm ws}(E) f_{\rm pkt} = \frac{E^2\Phi_{\rm q}(A_{\rm eff}\kappa_{\rm a}\tau_{\rm int})^2} {2N_{\rm b}B_{\rm sp}(2-B_{\rm sp})},
\label{eq:snr_prf_invariant}
\end{equation}
which is independent of $K_m$ under the stated assumptions.
\end{proposition}}

\begin{figure}[t]
\centering \includegraphics[width=2.8in]{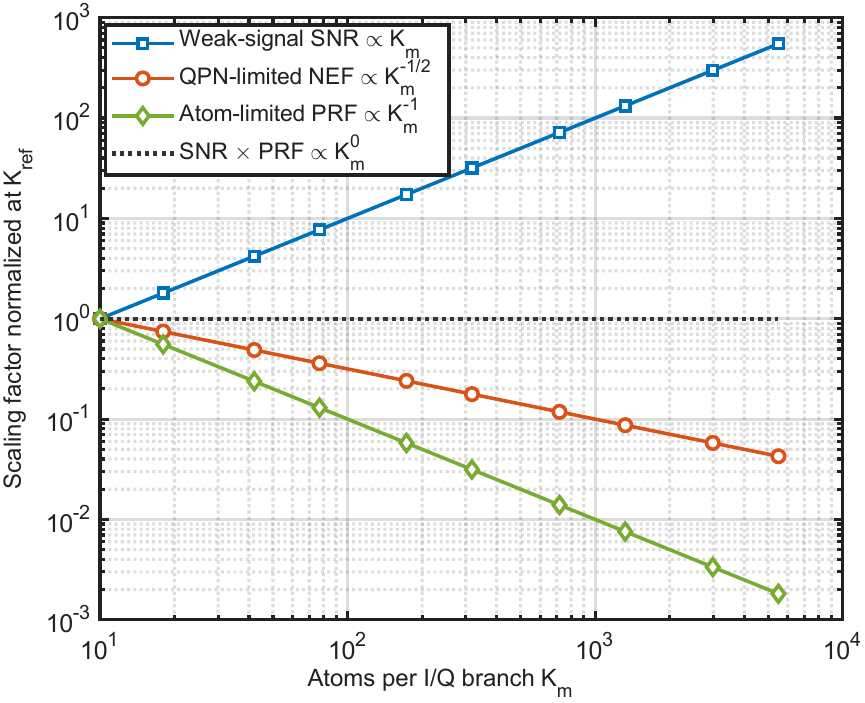}
\caption{{Sensitivity-sustainability scaling with active atom number $K_m$ per I/Q branch for the proposed CoRA.}}
\label{fig:information_prf_tradeoff}
\end{figure}

{Fig.~\ref{fig:information_prf_tradeoff} highlights the scaling law in Proposition~\ref{prop:sensitivity_sustainability}, where all quantities are normalized to their respective values at the reference atom number $K_{\rm ref}=10$. In the weak-field regime, the SNR increases approximately linearly with $K_m$ and the corresponding QPN-limited NEF decreases as $K_m^{-1/2}$. Once the initialized-atom flux is the active long-term constraint, the sustainable PRF decreases as $K_m^{-1}$. The product $\gamma_{\rm arr}^{\rm ws}f_{\rm pkt}$ is therefore independent of $K_m$ according to Proposition~\ref{prop:sensitivity_sustainability}. Thus, increasing the atom number redistributes a fixed atom-supply resource between instantaneous sensitivity and long-term packet sustainability rather than providing an unconstrained communication gain.}

\section{Numerical Results}
\label{sec:numerical_results}

In this section, we first evaluate the performance of the proposed {STM architecture and JBPS detection} for CoRA-based wireless communications, and then compare the specific communication metrics for the classical antenna, vapor-cell, and CoRA receivers.

\subsection{Simulation Setups}
\label{subsec:numerical_configuration_r6}

\begin{table}[t]
\centering
\caption{Communication-system and proposed-receiver parameters.}
\label{tab:reference_parameters}
\setlength{\tabcolsep}{3pt} \renewcommand{\arraystretch}{0.95}

\begin{tabular}{p{0.45\columnwidth}p{0.47\columnwidth}}
\hline \textbf{Parameter} & \textbf{Value} \\
\hline

Carrier frequency $f_c$ & $6.6~\mathrm{GHz}$ \\
Communication distance $d$ & $20~\mathrm{m}$ \\
Transmit antenna gain $G_t$ & $0~\mathrm{dBi}$ \\
Packet channel & Block Rician fading \\
Rician factor $K_R$ & $10~\mathrm{dB}$ \\
Shadowing deviation $\sigma_{\rm sh}$ & $2~\mathrm{dB}$ \\
Modulation Order & QPSK \\
Numerical samples per symbol & $32$ \\
Discrete atomic-gate gain $\alpha_g$ & $0.9645$ \\
Symbol interval $T_{\rm sym}$ & $1.0~\mu\mathrm{s}$ \\
Atomic interaction time $\tau_{\rm int}$ & $0.8~\mu\mathrm{s}$ \\
Field-to-Rabi coefficient $\kappa_E/(2\pi)$ & $4.88~\mathrm{kHz}/(1~\mu\mathrm{V/cm})$ \\
Reference Rabi scale $g_{\rm ref}$ & $1~\mathrm{rad/s}$ \\
Effective SPAM pair $(A_{\rm eff},B_{\rm sp})$ & $(0.438,0.680)$ \\
Population-loss rate $\Gamma_{\rm pop}/(2\pi)$ & $2~\mathrm{kHz}$ \\
Atoms per I/Q branch $K_m$ & $343$ \\
Burst length $N_{\rm b}$ & $16$ \\
Atom-supply flux $\Phi_q$ & $3\times10^4~\mathrm{s}^{-1}$ \\

Frequency-state correlation, $\rho_\nu$ & $0.985$ \\
Number of particles $N_{\rm PF}$ & $512$ \\
Number of smoothed paths $N_{\rm s}$ & $160$ \\
Resampling threshold $\alpha_{\rm res}$ & $0.55$ \\
Polarization projection $\zeta_{\rm pol}$ & $1$ \\
Scheduled I/Q-subarray pairs $M$ & $16$ \\
Payload-symbol prior $P_X(c)$ & $1/4,\;c\in\mathcal C$ \\
\hline
\end{tabular}
\end{table}

Table~\ref{tab:reference_parameters} summarizes the common link and proposed-receiver parameters. In the performance comparison for receivers with different front-end antennas, all front ends are compared using the same incident peak field, carrier frequency, and $B_{\rm n}=1~\mathrm{MHz}$ noise bandwidth. The classical receiver reference is a copolarized finite-gain antenna with $G_r=5.5~\mathrm{dBi}$, $\mathrm{NF}=6~\mathrm{dB}$, and $T_{\rm ant}=T_0=290~\mathrm{K}$. The vapor cell-based atomic receiver uses the same $^{87}\mathrm{Rb}$ $|68D_{5/2},5/2\rangle\leftrightarrow |69P_{3/2},3/2\rangle$ RF transition as the CoRA. We consider the classic four-level electronic ladder that is $5S_{1/2}\!\rightarrow5P_{3/2}\!\rightarrow68D_{5/2} \!\rightarrow69P_{3/2}$, with $780.24$- and $479.63$-nm optical fields and optical dipole magnitudes $\{2.98915,0.009815\}ea_0$. The common RF transition has the Alkali Rydberg Calculator (ARC) matrix element $3813.27ea_0$ \cite{vsibalic2017arc}, while all atomic receiver calculations use the experimentally calibrated $\kappa_E$ in Table~\ref{tab:reference_parameters}. The adopted $290$-K population-loss rates for $5P$ $68D$, and $69P$ are $2\pi\{6.0659~\mathrm{MHz}, 1.0775~\mathrm{kHz},0.8444~\mathrm{kHz}\}$, respectively. The vapor cell-based atomic receiver uses a $10$-cm cell, density $1.5\times10^{16}~\mathrm{m}^{-3}$, $290$-K temperature, a $30$-mW optical homodyne LO, $20$-kHz environmental dephasing, $50$-kHz linewidth per optical source, and a $1$-kHz RF-LO linewidth. The steady-state susceptibility is averaged over the one-dimensional Maxwell--Boltzmann velocity distribution before its field derivative and Liouvillian modulation bandwidth are evaluated. For each ladder, the center-frequency weak-signal NEF is minimized over $P_{\rm p}\in[10~\mu\mathrm W,2~\mathrm{mW}]$, $P_{\rm c}\in[1,160]~\mathrm{mW}$, $E_{\rm LO}^{\rm RF}\in[5,50]~\mathrm{mV/m}$, $\Delta\in[-5,5]~\mathrm{MHz}$, and $r_{\rm b}\in[0.25,0.80]~\mathrm{mm}$. The total optical-drive cap is $160~\mathrm{mW}$ and the probe-transmission floor is $0.01$.

In the performance evaluation of burst signal detection, we provide four baseline schemes under the same transmitted bursts, dynamic states, and binomial observations.
\begin{itemize}
\item \emph{Binomial MAP (BMAP)}: the composite channel-and-reference state is estimated from the leading acquisition pilots and then the phase and gain are frozen, with the phase increment set to zero. This baseline isolates the penalty of omitting within-burst tracking.
\item \emph{Extended Kalman filtering-based BMAP (EKF-BMAP)}: an EKF predicts the composite atom state at every symbol and updates it only at known pilot positions \cite{aghamohammadi1989adaptive}. BMAP then uses the resulting plug-in state.
\item \emph{Decision-directed EKF (DD-EKF)}: a confidence-gated payload decision updates the EKF for subsequent symbols \cite{mehrpouyan2012joint}. The current symbol is not redetected after its decision-directed update.
\item \emph{Oracle}: the true simulated state is inserted into BMAP. This is a nonimplementable conditional reference as the performance upper bound, not a practical receiver.
\end{itemize}

\subsection{Performance Evaluation for Burst Communications}

\begin{figure*}[t]
\centering
    \begin{minipage}[b]{0.32\textwidth}
\centering \includegraphics[width=\textwidth]{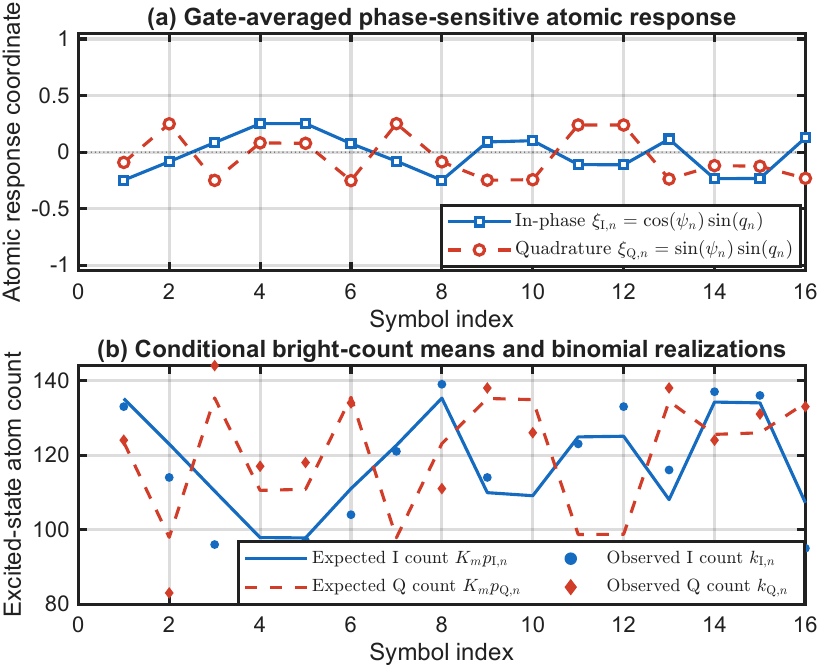}
\caption{Visualization for the reception of one burst packet.}
        \label{fig:end_to_end_chain_r6}
    \end{minipage}
\hfill
    \begin{minipage}[b]{0.32\textwidth}
\centering \includegraphics[width=\textwidth]{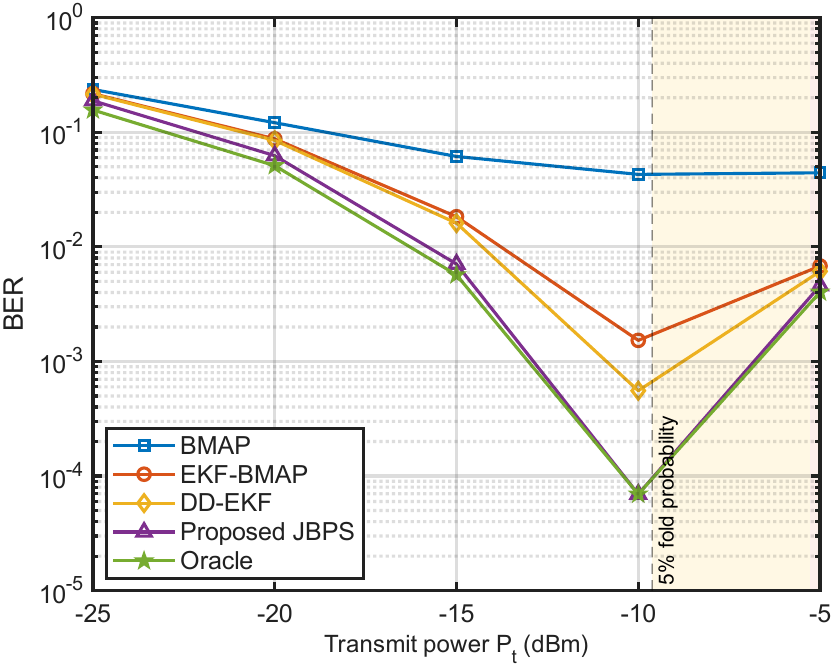}
\caption{BER versus transmit power for different signal detection algorithms.}
        \label{fig:power_ber_r6}
    \end{minipage}
\hfill
    \begin{minipage}[b]{0.32\textwidth}
\centering \includegraphics[width=\textwidth]{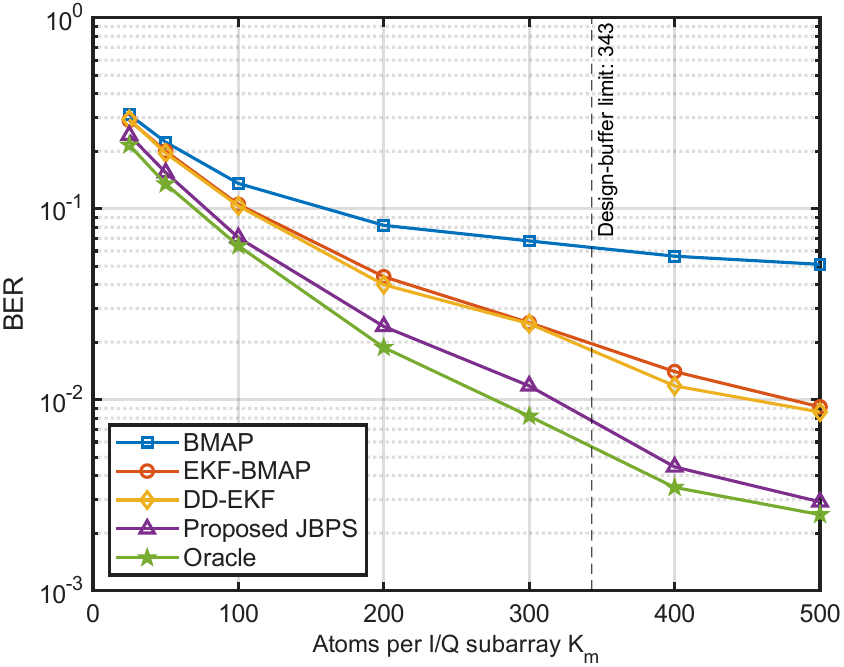}
\caption{BER versus atoms per I/Q subarray $K_m$ for different signal detection algorithms.}
        \label{fig:atom_ber_r6}
    \end{minipage}
    \label{fig:combined_group1}
\end{figure*}

Fig.~\ref{fig:end_to_end_chain_r6} visualizes the atomic observation generated by one burst, which connects the phase-sensitive quantum response to the count data actually available to the baseband detector. Fig.~\ref{fig:end_to_end_chain_r6}~(a) shows the noiseless response coordinates $\xi_{I,n}=\cos(\psi_n)\sin(q_n)$ and $\xi_{Q,n}=\sin(\psi_n)\sin(q_n)$ after the two analysis rotations. They are not conventional noisy I/Q voltage samples: the composite phase $\theta_n$ rotates the response pair, while the Rabi angle $q_n$ sets its nonlinear radial factor. Fig.~\ref{fig:end_to_end_chain_r6}~(b) maps these coordinates to the conditional means $K_mp_{I,n}$ and $K_mp_{Q,n}$ and compares them with the observed integer counts $k_{I,n}$ and $k_{Q,n}$. The symbol-to-symbol deviations of the counts from their conditional means are the finite-$K_m$ binomial fluctuations retained by the exact likelihood. For the illustrated $K_m=343$ and $N_{\rm b}=16$, a single-use packet requires $K_0=2K_mN_{\rm b}=10976$ initialized atoms that is comparable to the reported design-scale trapping capacity \cite{wang2026trapping}.

Fig.~\ref{fig:power_ber_r6} evaluates BER as a function of transmit power for different signal detection algorithms. In the low-power region, all detection schemes are primarily limited by finite-count uncertainty, and increasing $P_t$ improves the distinguishability of the binomial hypotheses. BMAP improves much more slowly because its frozen acquisition state cannot follow the accumulated phase drift. EKF-BMAP and DD-EKF reduce this mismatch, but both still make detection from a point estimate of the latent state. The proposed JBPS instead retains multiple state hypotheses and uses future burst observations through smoothing. Consequently, it approaches the Oracle curve and provides substantially lower BER than the frozen-state and plug-in baselines. The high-power BER upturn originates from the periodic Rabi transduction of the atom-array receiver. For the considered I/Q measurement, the radial response of the two binomial branches is proportional to $\sin(q_n)$, where $q_n=\kappa_{\rm a}|\bar{E}_n|\tau_{\rm int}$ is the instantaneous Rabi angle referred to the common incident-field plane. Within the first response lobe, i.e., $0\leq q_n<\pi/2$, increasing the incident field increases $\sin(q_n)$ monotonically and therefore improves the separation between the bit-conditioned count distributions. The point $q_n=\pi/2$ is referred to as the {first Rabi fold}, because the response slope $ \frac{\partial\sin(q_n)}{\partial q_n}=\cos(q_n)$ vanishes at this point and changes sign beyond it. Consequently, different Rabi angles, such as $q_n$ and $\pi-q_n$, can produce the same atomic response. Additional folds occur because $\sin(q_n)$ is periodic. Under block fading, the transmit power at which the first fold is reached is itself random. In particular, the vertical guide marks that $5\%$ of the channel realizations drive the maximum-radius symbol beyond the first Rabi fold. The shaded region indicates a progressively larger fold probability and is a dynamic-range constraint.

Fig.~\ref{fig:atom_ber_r6} presents BER versus atoms per I/Q subarray $K_m$ for different signal detection algorithms. Increasing $K_m$ reduces the normalized binomial variance, so the Oracle and JBPS curves continue to improve. The acquisition-only BMAP curve presents a much higher error floor because increasing the atom number suppresses observation noise but does not correct its frozen-state mismatch. The EKF-based receivers fall between these cases: their residual error is increasingly determined by state-estimation error rather than QPN. The proposed JBPS remains closest to the Oracle because it marginalizes payload uncertainty during filtering and smooths the state using the entire burst. The vertical line at $K_m=343$ marks the design-scale $K_0=10976$ single-use buffer boundary for the adopted $N_{\rm b}=16$ burst. Values to its right require a larger initialized buffer or explicit atom reinitialization/reuse and are not claimed as feasible single-use bursts under the same $K_0$.

\subsection{Performance Comparison for Different Receivers}

\begin{figure*}[t]
\centering
    \begin{minipage}[b]{0.32\textwidth}
\centering \includegraphics[width=\textwidth]{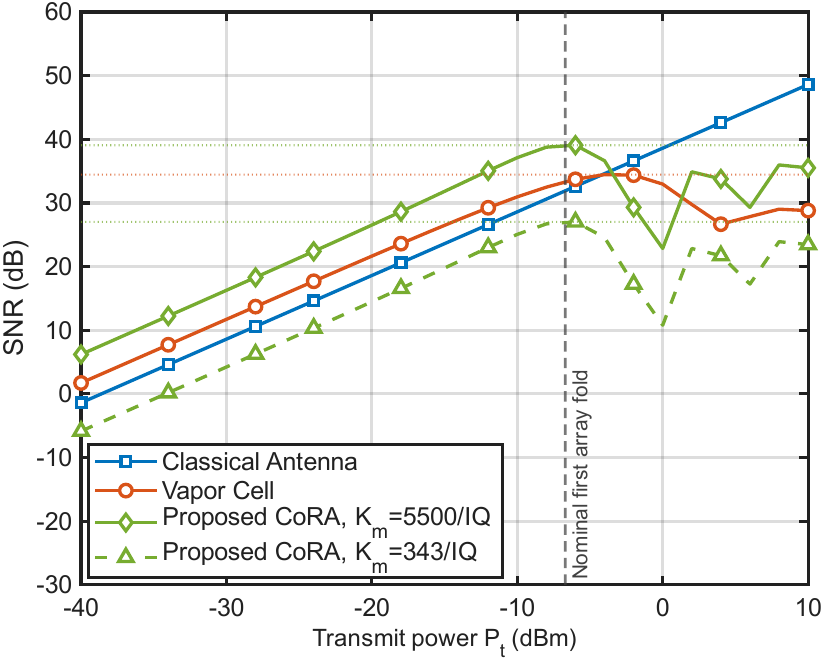}
\caption{SNR versus transmit power for different receivers.}
        \label{fig:strong_field_response}
    \end{minipage}
\hfill
    \begin{minipage}[b]{0.32\textwidth}
\centering \includegraphics[width=\textwidth]{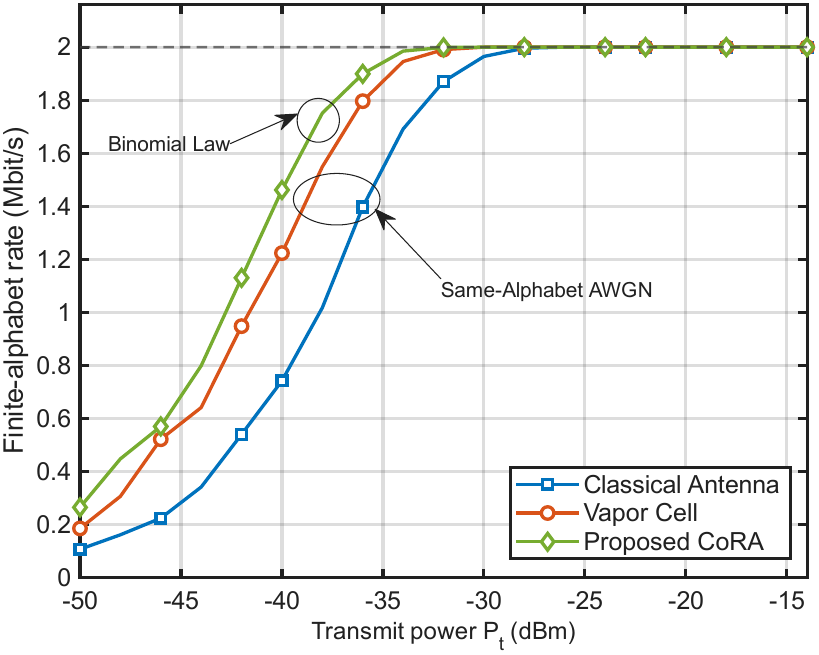}
\caption{Finite-alphabet rate versus transmit power for different receivers.}
        \label{fig:finite_alphabet_information}
    \end{minipage}
\hfill
        \begin{minipage}[b]{0.32\textwidth}
\centering \includegraphics[width=\textwidth]{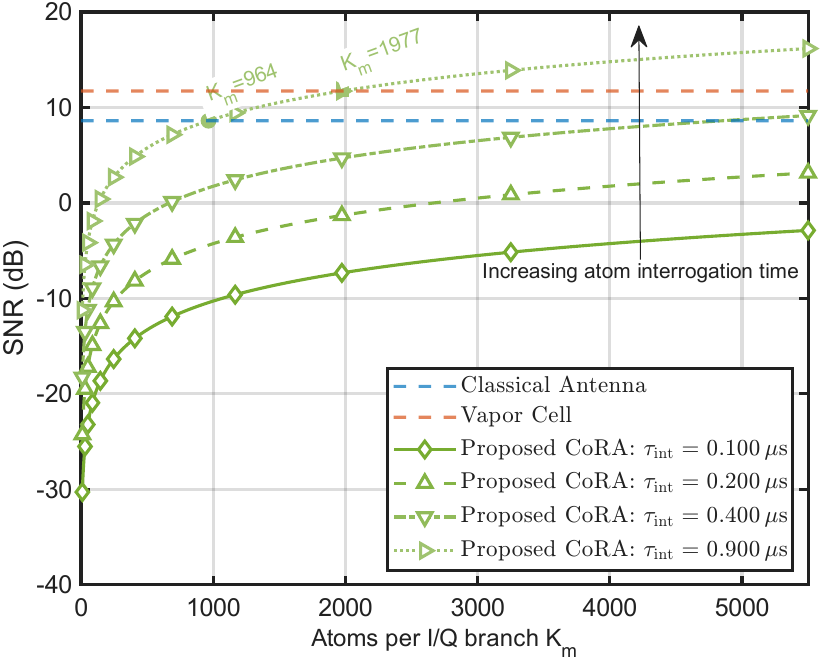}
\caption{SNR versus atoms per I/Q branch $K_m$ for different interrogation times $\tau_{\rm int}$.}
        \label{fig:integration_time_sensitivity}
    \end{minipage}

    \label{fig:combined_group2}
\end{figure*}

Fig.~\ref{fig:strong_field_response} compares the SNR performance versus transmit power for different receivers. In the weak-field region, all curves follow their locally linear trends. The array with $K_m=5500$ atoms per I/Q branch has the highest SNR because its finite-count uncertainty is lower, whereas the single-use $K_m=343$ array does not universally outperform the classical and vapor references. This distinction prevents an ideal large-array limit from being presented as the performance of the currently available packet buffer. The $K_m=5500$ curve denotes that one simultaneous I/Q measurement uses $11000$ atoms, showing the sensitivity potential of the large-scale atom array. The three receiver architectures exhibit fundamentally different strong-field behaviors in Fig.~\ref{fig:strong_field_response}. The classical receiver is kept linear over the displayed range because the adopted benchmark assumes an unsaturated antenna and RF front end. For the CoRA, the exact finite-count SNR retains the periodic factor $\sin(q)$. As $P_t$ increases, the atomic response first grows over $0\leq q\leq\pi/2$ and reaches its first-lobe maximum at $q=\pi/2$. Beyond this point, $\sin(q)$ folds back, becomes zero at $q=\pi$, and subsequently enters additional positive and negative Rabi lobes. The count-domain SNR therefore exhibits repeated peaks and nulls instead of continuing along its weak-field linear extrapolation. The vertical line in Fig.~\ref{fig:strong_field_response} denotes the nominal no-fading first-fold condition $\kappa_{\rm a}E_0(P_t,d)\tau_{\rm int} \max_{c\in\mathcal C}|c| =\frac{\pi}{2}.$ The vapor cell-based receiver employs a fixed-microwave-LO atomic-superheterodyne architecture. At larger signal fields, the signal is no longer a small perturbation to the fixed LO. The resulting change in atomic coherence causes response compression and eventual SNR roll-off. Accordingly, the results demonstrate both the weak-field sensitivity benefit of the atom array and the need for dynamic-range management at high incident powers.

Fig.~\ref{fig:finite_alphabet_information} compares the derived state-averaged finite-alphabet rate in Section~\ref{subsec:finite_alphabet_rate} for three receiver architectures under the same QPSK input alphabet and symbol interval. For the CoRA, the rate is obtained from the exact I/Q binomial transition law, whereas the classical-antenna and hot-vapor-cell benchmarks use equivalent complex-AWGN channels parameterized by their weak-signal SNRs. As the transmit power increases, the symbol-conditioned output distributions become more distinguishable, and the finite-alphabet rate correspondingly increases toward the common QPSK ceiling $H(X)/T_{\rm sym}=\log_2|\mathcal C|/T_{\rm sym}=2~\mathrm{Mbit/s}$ for $T_{\rm sym}=1~\mu\mathrm{s}$. Within the considered weak-field region, the CoRA reaches a given rate at a lower transmit power under the adopted receiver parameters, while the common high-SNR saturation is imposed by the finite QPSK alphabet rather than by identical noise statistics across the three receiver architectures.

Fig.~\ref{fig:integration_time_sensitivity} isolates the interaction-time dependence at $P_t=-30~\mathrm{dBm}$ for cold-atom arrays. In the local weak-field regime, Eq.~\eqref{eq:lemma1_snr} gives approximately $ \gamma_{\rm arr}^{\rm ws} \propto K_m A_{\rm eff}^{\,2}\tau_{\rm int}^{\,2},$ for fixed $B_{\rm sp}$ and incident field. A longer interaction therefore reduces the number of atoms required to cross a prescribed SNR. For $\tau_{\rm int}=0.9~\mu\mathrm{s}$, the annotated crossings occur at approximately $K_m=964$ for the classical reference and $K_m=1977$ for the optimized vapor reference. By contrast, the $0.1$- and $0.2$-$\mu\mathrm{s}$ curves do not reach these references within the plotted atom range. The gain from a longer interaction is not free, where it lowers the first-fold field, reduces the available foreground-control time, and constrains the symbol period.

\section{Conclusion}
\label{sec:conclusion}

In this paper, we investigated CoRA assisted wireless communication systems, where the practical operation is limited by destructive state-selective readout, nonlinear Rabi responses, and unknown phase-and-gain evolution. To address these challenges, we proposed an STM architecture that schedules preloaded atom subarrays over successive symbol intervals to enable gap-free interrogation within an admitted burst. Further, we proposed a JBPS detector that retains the exact binomial likelihood and averages bit decisions over phase-and-gain trajectories smoothed using the entire burst. We derived the weak-signal sensitivity, receiver-crossover conditions and finite-alphabet rate. Numerical results demonstrate that the proposed STM architecture effectively resolves the macroscopic temporal bottleneck for in-burst continuous sampling, while the proposed JBPS achieves superior BER performance over the considered baselines. \textbf{\textit{We revealed a sensitivity-sustainability scaling law: in the QPN-limited and atom-flux-limited regime, increasing the active atom number $K_m$ linearly improves the weak-signal SNR but reduces the sustainable packet repetition frequency as $K_m^{-1}$.}} Future work will exploit parallel subarray operation with pipelined atom replenishment to enhance the long-term sustainability of CoRA systems.

\appendices

\section{Proof of Lemma~\ref{lem:array_weak_field_nef}}
\label{app:proof_lemma1}

For the matched I/Q branches, the excitation probabilities can be expressed as
\begin{align}
p_I(E,\psi) &=\frac{B_{\rm sp}}{2} -\frac{A_{\rm eff}}{2}\cos\psi\sin q,\nonumber\\
p_Q(E,\psi) &=\frac{B_{\rm sp}}{2} -\frac{A_{\rm eff}}{2}\sin\psi\sin q,
\label{eq:app_branch_probabilities_compact}
\end{align}
where $q=\kappa_{\rm a}E\tau_{\rm int}$.

Since the two binomial counts are conditionally independent, their Fisher-information contributions add. The field-amplitude Fisher information is given by
\begin{equation}
J_E(E,\psi) =K_m\sum_{b\in\{I,Q\}} \frac{[\partial p_b(E,\psi)/\partial E]^2} {p_b(E,\psi)[1-p_b(E,\psi)]}.
\label{eq:app_fisher_sum_compact}
\end{equation}
Using $\sin q=q+\mathcal O(q^3)$, the required quantities at the zero-field operating point are given by
\begin{align}
&p_I(0,\psi)=p_Q(0,\psi) =\frac{B_{\rm sp}}{2},\nonumber\\
&\left.\frac{\partial p_I}{\partial E}\right|_{E=0} =-\frac{A_{\rm eff}\kappa_{\rm a}\tau_{\rm int}}{2}\cos\psi,\nonumber\\
&\left.\frac{\partial p_Q}{\partial E}\right|_{E=0} =-\frac{A_{\rm eff}\kappa_{\rm a}\tau_{\rm int}}{2}\sin\psi.
\label{eq:app_weak_field_terms_compact}
\end{align}
Substitution into Eq.~\eqref{eq:app_fisher_sum_compact} gives
\begin{align}
J_E^{\rm ws} &=\frac{K_m(A_{\rm eff}\kappa_{\rm a}\tau_{\rm int})^2} {B_{\rm sp}(2-B_{\rm sp})} \left(\cos^2\psi+\sin^2\psi\right)\nonumber\\
&=\frac{K_m(A_{\rm eff}\kappa_{\rm a}\tau_{\rm int})^2} {B_{\rm sp}(2-B_{\rm sp})},
\label{eq:app_fisher_result_compact}
\end{align}
which proves Eq.~\eqref{eq:lemma1_fisher} and its leading-order independence from the tracked phase. Finally, $\gamma_{\rm arr}^{\rm ws}=E^2J_E^{\rm ws}$ and $1/J_E^{\rm ws}=(\mathcal E_{\rm arr}^{\rm QPN})^2B_{\rm n}$ give Eqs.~\eqref{eq:lemma1_snr} and~\eqref{eq:lemma1_nef}, respectively.

\bibliographystyle{IEEEtran}
\bibliography{IEEEabrv,refs_RAA}

\end{document}